%% file: paper.tex
\documentclass[11pt, leqno]{article}
\usepackage[margin=1in]{geometry}

\usepackage{hyperref}
\hypersetup{colorlinks=true,linkcolor=blue,citecolor=blue,urlcolor=blue}

\usepackage[utf8]{inputenc}
\usepackage{amsfonts}
\usepackage{graphicx}
\usepackage{amsthm}
\usepackage{mathrsfs}
\usepackage{amssymb}
\usepackage{hyperref}
\usepackage{amsmath}
\usepackage[backend=biber,style=alphabetic]{biblatex}
\usepackage{graphicx}
\usepackage{xcolor}
\usepackage{bbm}
\usepackage{thmtools}
\usepackage{thm-restate}
\usepackage{hyperref}
\usepackage{cleveref}
\usepackage{xspace}
\usepackage{comment}
\usepackage{tabularx}
\usepackage{tikz}
\usepackage{caption}
\usepackage{subcaption}
\usepackage{wrapfig}
\usepackage{mdframed}
\usepackage[many]{tcolorbox}

\newtheorem{theorem}{Theorem}
\newtheorem*{theorem*}{Theorem}
\newtheorem{proposition}{Proposition}
\newtheorem{lemma}{Lemma}[section]
\newtheorem{definition}[lemma]{Definition}

\newtheorem{corollary}[lemma]{Corollary}

\newtheorem{notation}[lemma]{Notation}

\declaretheoremstyle[
	headfont		= \fontfamily{ptm}\selectfont\bfseries,			%
	headformat		= \NAME~\NUMBER\NOTE\vspace{-3mm},	%
	headpunct		= \newline\newline,				%
	notefont		= \fontfamily{ptm}\selectfont,						%
	notebraces		= {--- }{},						%
	bodyfont		= \fontfamily{ptm}\selectfont\small,			%
	postheadspace	= 10pt,							%
 preheadspace = 2pt
	]{algorithm}

\theoremstyle{algorithm}
\newtheorem{algorithm}{Algorithm}

\tcolorboxenvironment{algorithm}{
  sharp corners,
  boxrule=1pt,
  colback=white,
  before skip=2\topsep,
  after skip=2\topsep,
  center,
  width=0.85\linewidth
}

\DeclareMathOperator*{\argmax}{arg\,max}

\DeclareMathOperator*{\argsup}{arg\,sup}

\newcommand{\nice}{simple\xspace}
\newcommand{\Nice}{Simple\xspace}

\newcommand{\textib}[1]{\textit{\textbf{#1}}}
\newcommand{\FBV}{\textup{OPT}(\boldsymbol{V})}%
\newcommand{\SBV}{\textup{OPT}_\textup{BIC}(\boldsymbol{V})}%
\newcommand{\OPTniceV}{\textup{OPT}_\textup{\nice}(\boldsymbol{V})}

\newcommand{\OLD}[1]{}

\makeatletter
\newcommand{\multiline}[1]{%
  \begin{tabularx}{\dimexpr\linewidth-\ALG@thistlm}[t]{@{}X@{}}
    #1
  \end{tabularx}
}
\makeatother

\AtBeginEnvironment{align}{\setcounter{equation}{0}}
\AtBeginEnvironment{gather}{\setcounter{equation}{0}}

\begin{document}

\title{Learning to Maximize Gains From Trade in Small Markets\footnote{First version: January 2024. Supported by the European Research Council (ERC) under the European Union’s Horizon 2020 Research and Innovation Programme (grant agreement no. 740282) and by a grant from the Israeli Science Foundation (ISF number 505/23). Moshe Babaioff's research is supported in part by a Golda Meir Fellowship.}}

\author{Moshe Babaioff\footnote{Hebrew University of Jerusalem. E-mail: \texttt{moshe.babaioff@mail.huji.ac.il}} \and Amitai Frey\footnote{Hebrew University of Jerusalem. E-mail: \texttt{amitai.frey@mail.huji.ac.il}}  \and Noam Nisan\footnote{Hebrew University of Jerusalem. E-mail: \texttt{noam@cs.huji.ac.il}} }
\date{}

\maketitle

\thispagestyle{empty}
\begin{abstract}

    We study the problem of designing a two-sided market (double auction) to maximize the gains from trade (social welfare) under the constraints
    of (dominant-strategy) incentive compatibility and budget-balance.  
    Our goal is to do so for an
    unknown distribution from which we are given a polynomial number of samples.  Our first result is a
    general impossibility for the case of correlated distributions of values even between
    just one seller and two buyers, in contrast to the case of one seller and one buyer 
    (bilateral trade) where this is possible.  Our second result is an efficient
    learning algorithm for one seller and two buyers in the case of independent distributions which is based on a novel algorithm for computing 
    optimal mechanisms for finitely supported and explicitly given independent distributions.  Both results rely heavily on characterizations of 
    (dominant-strategy) incentive compatible mechanisms that are strongly budget-balanced.
\end{abstract}

\pagenumbering{arabic} 

\input{intro}

\section{Model and Notation}\label{sec:model}

Consider a market with one seller and two buyers, where the seller has one indivisible item to sell, and each buyer wants to buy that item. We denote the seller by $s$, and the two buyers by $b_1$ and $b_2$. The private value of the seller for keeping the item is $v_s\geq 0$, and the private value of buyer $b_i\in \{b_1,b_2\}$ for getting the item is $v_i\geq 0$. The triplet of values  $\textib{v}=(v_s,v_1,v_2)\in \mathbb{R}_{\geq 0}^3$ is sampled from a joint prior distribution $\textib{V}$ which is supported on a subset of $\mathbb{R}_{\geq 0}^3$. 
When referring to vectors of valuations, we  sometime use the standard notation $(v_i,\textib{v}_{-i})$ to separate the valuation of agent $i$ from the others. When referring to functions relating to buyer $b_i$, we sometimes omit the $b$ and refer to them as $f_i$, meaning $f_{b_i}$. Also, when discussing a buyer $b\in \{b_1,b_2\}$ we sometimes refer to the other buyer as $-b$.

\subsection{\Nice Mechanisms}

In this section we formally define the ``\nice'' mechanisms that we study.  
Basically, we focus on (dominant strategy) truthful mechanisms where the buyer who gets the item -- if trade occurs -- pays the seller, and no payments occur for non-traders.  The basic example of such a mechanism is where the sale price is
determined by bid of the {\em low buyer}, and trade happens with the high bidder 
if the seller's bid
is lower than the price.  
A very different type of mechanism of this type is a sequential one, where the first buyer gets a chance to trade at some fixed price,
and only if the first buyer declines then the second buyer gets a chance to 
trade
at some other fixed price.  As we will see, however, 
the full space of \nice mechanisms is much richer,  %
and contains significant generalizations.

A deterministic (direct revelation) mechanism $M=(A,p)$ for a 1-seller 2-buyer setting is constructed from an allocation rule $A:\mathbb{R}_{\geq 0}^3 \rightarrow \{s,b_1,b_2\}$  and a payment function $p:\mathbb{R}_{\geq 0}^3 \rightarrow \mathbb{R}^3$.
In such a direct revelation mechanism, agents are asked to report their valuations (submit bids), and we denote the profile of reports  by $\textib{v}'=(v_s',v_1',v_2')\in \mathbb{R}_{\geq 0}^3$.
Given reports $\textib{v}'= (v'_s,v'_1,v'_2)$,  the allocation $A(\textib{v}')=A(v'_s,v'_1,v'_2)$ denotes the agent which ends up with the item, and $p_i(\textib{v}')= p_i(v'_s,v'_1,v'_2)$  denotes the payment paid by agent $i$.\footnote{If a payment is negative, that agent receives a payment.}

Denote a trade between the seller and buyer $b \in \{b_1,b_2\}$ by $s\rightarrow b$. Let $M=(A,p)$ be a mechanism, we define the utilities of the agents with respect to the mechanism $M$ when the values are $\textib{v}$  and the reports are $\textib{v}'$ as follows:
\begin{itemize}
    \item Seller $s$: $u_s(v_s,\textib{v}')=-\mathbbm{1}\{A(\textib{v}') \neq s\}\cdot v_s-p_s(\textib{v}')$.
    \item Buyer $b \in \{b_1,b_2\}$: $u_b(v_b,\textib{v}')=\mathbbm{1}\{A(\textib{v}') = b\}\cdot v_b - p_b(\textib{v}')$.
\end{itemize}

In a DSIC mechanism it is a dominant strategy for each agent to bid truthfully: 
\begin{definition}
    A mechanism $M$ is \textup{Dominant Strategy Incentive Compatible} (DSIC) if for every agent $i\in \{s,b_1,b_2\}$, for every valuation profile $\textib{v}=(v_s,v_1,v_2)$ and every bidding profile $\textib{v}'=(v_s',v_1',v_2')$ it holds that
    $$u_i(v_i,(v_i,\textib{v}_{-i}')) \geq u_i(v_i,(v_i',\textib{v}_{-i}'))$$ 
\end{definition}

From now on, whenever a mechanism is DSIC we assume that all the agents are truthful, as being truthful is a dominant strategy for every agent. When truthfulness also ensures non-negative utility, participation in the mechanism is individually rational: 

\begin{definition}
    A mechanism $M$ is \textup{ex-post Individually Rational} (IR) if for every agent $i\in \{s,b_1,b_2\}$, every valuation profile $\textib{v}=(v_s,v_1,v_2)$ and every bidding profile $\textib{v}'=(v_s',v_1',v_2')$: 
        \begin{gather*}
            u_i(v_i,(v_i,\textib{v}_{-i}')) \geq 0
    \end{gather*}
\end{definition}

A natural requirement is that the mechanism never loses money. We focus on mechanisms that also never gain any money:
\begin{definition}
    A mechanism $M$ is \textup{Strongly Budget-Balanced} (SBB) if:
    \begin{gather*}
        \forall v'_s,v'_1,v'_2: \sum_{i\in \{s,b_1,b_2\}}p_i(v'_s,v'_1,v'_2)=0
    \end{gather*}
\end{definition}

We say that an agent is \emph{winning} if she is involved in a trade. 
\begin{definition}[Winning Outcomes]
    The \textup{Winning Outcomes} $W_i\subseteq \{s,b_1,b_2\}$ of agent $i$ are a set of all the outcomes in which agent $i$ is involved in a  trade. Specifically:
    \begin{align*}
        W_1&=\{1\} & W_2&=\{2\} & W_s&=\{1,2\}
    \end{align*}
\end{definition}

As usual, we focus on mechanisms that are normalized (losers pay zero)\footnote{Note that this property does not necessarily hold for all ex-post IR and SBB mechanisms. E.g., a mechanism that offers a trade $b\rightarrow 1$ with $p_1=1, p_2=-1/2, p_s=1/2$. If the agents agree to the trade, Buyer 2 is a loser but his price is not 0.} :
\begin{definition}
    A mechanism $M=(A,p)$ is called \textup{normalized} if the payment when losing is always 0, i.e. for every agent $i$ and for every bidding profile $\textib{v}'=(v_s',v_1',v_2')$, if $A(\textib{v}')\notin W_i$ then $p_i(\textib{v}')=0$.
\end{definition}

In this paper we focus on \nice mechanisms: a mechanism is \emph{\nice} if it is deterministic, normalized, DSIC, ex-post IR and SBB.
{As our setting is a single-parameter domain, the standard general characterization of normalized DSIC mechanisms (monotone allocation with critical value payments) applies to it.
We next present its implications for the model of a seller and two buyers. }

First, we define allocation monotonicity. An allocation rule is monotone if a winner does not become a loser by submitting a better bid: 

\begin{definition}
    An allocation function $A$ is monotone if:
    \begin{enumerate}
        \item $\forall (v_s',v_1',v_2'), \forall b\in \{b_1,b_2\},  \forall v_b'' \geq v_b':A(v_b',\textib{v}_{-b}')\in W_b \Rightarrow A(v_b'',\textib{v}_{-b}')\in W_b$. 
        \item $\forall (v_s',v_1',v_2'), \forall v_s''\leq v_s' :A(v_s',v_1',v_2')\in W_s \Rightarrow A(v_s'',v_1',v_2')\in W_s$.
    \end{enumerate}
\end{definition}

A celebrated result by Myerson \cite{myerson1981optimal} implies a characterization of normalized DSIC mechanisms in a single-parameter domain. We will use a simple statement of this result, based on Theorem 9.36 in \cite{agt}, as it applies for our domain:

\begin{proposition} \label{prop:classicmech}
    A normalized mechanism $M=(A,p)$ is DSIC if both of the following conditions hold:
    \begin{enumerate}
        \item $A$ is monotone.
        \item The payment functions define the critical values for winning: 
        \begin{enumerate}
            \item For every buyer $b$, if $v_b' > p_b(\textbf{v}')$ then he wins,   and if $v_b' < p_b(\textbf{v}')$ he loses.
            \item For the seller, if $v_s' < -p_s(\textbf{v}')$ then she wins, and if $v_s' > -p_s(\textbf{v}')$ she loses.
        \end{enumerate}    
    \end{enumerate}
\end{proposition}

\subsection{Gains From Trade}
In this paper we focus on optimizing the gains form trade (GFT).  
Given a joint prior distribution $\textib{V}$ over triplet of values  $\textib{v}=(v_s,v_1,v_2)\in \mathbb{R}_{\geq 0}^3$, the \emph{unconstrained optimal GFT} (first best), denoted by $\FBV$, is defined to be:
$$ \FBV = \mathbb{E}_{\textib{v}\sim \textbf{V}}\left[\max \{v_1-v_s,v_2-v_s,0\}\right]$$ 

We are interested in the gains form trade of \nice mechanisms. 
We next establish some useful notations.
\begin{definition}[(Truthful) Gains From Trade]
    Let $M=(A,p)$ be a 1-seller 2-buyer DSIC mechanism. %
    The \textup{(Truthful) Gains From Trade (GFT)} of $M$ on valuation profile  $\textib{v}=(v_s,v_1,v_2)$ is defined to be:
    \begin{gather*}
        \textup{GFT}(M,\textib{v}) = \sum_{b\in\{1,2\}}(v_b-v_s)\cdot \mathbb{I}\{A(\textib{v}) =b\}
    \end{gather*}
\end{definition}
Using this notation we can now define the expected GFT of a mechanism for a given prior. 

\begin{definition}[Expected (Truthful)  GFT]
    Let $M=(A,p)$ be a 1-seller 2-buyer DSIC mechanism, and let $\textbf{V}$ be a joint prior over the agents' valuations. The  \textup{Expected Gains From Trade} of $M$ on prior $\textbf{V}$ is defined to be:
    \begin{gather*}
        \textup{GFT}(M,\textbf{V}) = \mathbb{E}_{\textib{v}\sim \textbf{V}}\left[\textup{GFT}(M,\textib{v})\right]
    \end{gather*}
\end{definition}

\section{Characterization of \Nice 1-Seller 2-Buyer Mechanisms}\label{sec:char}

A 1-seller 2-buyer mechanism $M=(A,p)$ has payment functions that depend on the reports of all agents, yet when the mechanism is \nice, this imposes significant constraints on these functions. These constraints will allow us to 
represent \nice mechanisms in less complicated fashion, 
by observing that any such mechanism can be associated with just two (single-parameter) payments functions, one for each buyer: each of these functions is only a function of one parameter, the report of the other buyer. 

This pair of functions determines the payments for all three agents whenever there is trade. Such a pair of functions must satisfy a condition that we call \emph{compatibility}, essentially stating that, as these functions define critical values for winning, it can never be the case that they imply that both buyers win together.  Thus, we show that any \nice mechanism is associated with a unique pair of (single-parameter) payments functions that are compatible. We also show that any compatible pair of single-parameter payment functions can be used to construct a \nice mechanism that is associated with these functions, and it maximizes the GFT over all such mechanisms.

\subsection{Associated  Functions}

We next show that any \nice mechanism can be associated with two (single-parameter) payments functions, $f_1, f_2 : \mathbb{R}_{\geq 0} \rightarrow \mathbb{R}_{\geq 0}\cup \{\infty\}$ 
such that this pair of functions $(f_1,f_2)$ determines the payments for all three agents whenever there is trade.  

\begin{definition}\label{def:payfuncs}
    Let $M=(A,p)$ be a \nice mechanism, we define $M$'s \textup{associated pair of functions} $(f_1,f_2)$ to be the pair of single-parameter functions $f_1,f_2 : \mathbb{R}_{\geq 0} \rightarrow \mathbb{R}_{\geq 0} \cup \{\infty\}$ defined as follows. For $v_1,v_2\geq 0$:
    \begin{gather*}
        f_1(v_2) = \begin{cases}
            p_1(v'_s,v'_1,v_2), & \exists v'_s,v'_1: A(v'_s,v'_1,v_2) = b_1 \\
            \infty, & \text{otherwise}
        \end{cases} \\ f_2(v_1) = \begin{cases}
            p_2(v'_s,v_1,v'_2), & \exists v'_s,v'_2: A(v'_s,v_1,v'_2) = b_2 \\
            \infty, & \text{otherwise}
        \end{cases}
    \end{gather*}
    
\end{definition}

We argue that these functions are indeed well defined (e.g., $f_1(v_2)$ does not depend on the choice of  $v'_s$ and $v'_1$). 
As $M$ is a DSIC mechanism, by the Taxation Principle, the payment of each buyer when winning cannot be affected by his own report. Additionally, as the mechanism is SBB and normalized, when there is a trade, the seller receives the payment of the trading buyer, and thus the trading buyer's payment cannot depend on the report of the seller. Therefore, there is a single parameter function that defines the payment in such a trade, and it can only depend on the non-trading buyer's value. 

We define the value of the function above for buyer $b$, to be the payment she is 
charged when buying the item, assuming that buyer $-b$ reported $v_{-b}$, and both the seller and buyer $b$ reported some hypothetical values that results with $b$ buying the item. The properties we mentioned indicate that the payment remains consistent in this hypothetical scenario, as in the actual case under consideration. %
This idea is what leads the functions above to be well-defined; for the sake of completeness we present a proof of these properties in Appendix \ref{app:missing31}.

Such a pair of functions must satisfy a condition that we call \emph{compatibility}:

\begin{definition}[Compatible Functions]
    Let $f_1,f_2 : \mathbb{R}_{\geq 0} \rightarrow \mathbb{R}_{\geq 0}\cup \{\infty\}$ be functions. The pair of functions $(f_1, f_2)$ is called \textup{compatible} if there \textbf{do not} exist $v_1, v_2\in \mathbb{R}_{\geq 0}$, such that $(v_1 > f_1(v_2)) \wedge (v_2 > f_2(v_1))$.
\end{definition}

Intuitively, the functions $f_1$ and $f_2$ represent the critical values for the buyers to trade, so they can never imply that both buyers trade concurrently. 
The compatibility (``non-crossing") condition does exactly that - it disallows functions that imply that both buyers are each bidding strictly above his critical value.  

Thus, we can associate every \nice mechanism with a compatible pair of (single-parameter) payment functions: 
\begin{lemma}[Associated Functions]\label{lem:charfunc}
    Let $M=(A,p)$ be a \nice mechanism, then there exists a unique pair of compatible functions $ (f_1,f_2)$
    s.t. for all $\textib{v}=(v_s,v_1,v_2)$, if $A(\textib{v})=b_1$ then $p(\textib{v}) = (-f_1(v_2),f_1(v_2),0)$, and if $A(\textib{v})=b_2$, then $p(\textib{v}) = (-f_2(v_1),0,f_2(v_1))$. We say that  $(f_1,f_2)$ is \emph{the pair of functions associated} with the mechanism $M$.
\end{lemma}

The proof of the lemma is deferred to Appendix \ref{app:missing31}.

\subsection{The Best Mechanism for Compatible Functions}

We have shown above that every \nice mechanism is associated with a compatible pair of functions. We next show the other direction - any pair of compatible functions can be used to construct a \nice mechanism such that the mechanism is associated with these functions, and moreover, it has the highest GFT of all such mechanisms. We first define this mechanism and then in Lemma \ref{lem:compmech} we show that it is \nice and associated with the given pair of functions.

\begin{definition}\label{def:m}
Let $f_1,f_2 : \mathbb{R}_{\geq 0} \rightarrow \mathbb{R}_{\geq 0}\cup \{\infty\}$ be a pair of functions. If the pair $(f_1, f_2)$ 
is compatible then there exists a mechanism $M_{f_1,f_2}=(A,p)$, which we call the \emph{mechanism corresponding to $(f_1, f_2)$}, that is well-defined and  %
is defined as follows.  
Given  a bidding profile $\textib{v}'=(v_s',v_1',v_2')$ the allocation and payments of $M_{f_1,f_2}$ are as follows:
\begin{enumerate}
    \item If $v_1' > f_1(v_2') \geq v_s'$: $A(\textib{v}') = b_1,\: p_1(\textib{v}') = -p_s(\textib{v}') = f_1(v_2'),\: p_2(\textib{v}') = 0$.
    \item Else, if $v_2' > f_2(v_1') \geq v_s'$: $A(\textib{v}') = b_2,\: p_2(\textib{v}') = -p_s(\textib{v}') = f_2(v_1'),\: p_1(\textib{v}') = 0$.
    \item Else, if $(v_1' = f_1(v_2') \geq v_s') \wedge (v_2' = f_2(v_1') \geq v_s'$): %
    $$for\ b=\begin{cases}
        b_1 & v_1' \geq v_2' \\
        b_2 & v_1' < v_2'
    \end{cases}, let\:A(\textib{v}')=b,\: p_b(\textib{v}') = -p_s(\textib{v}') = v_b',\: p_{-b}(\textib{v}') = 0$$
    \item Else, if $v_1' = f_1(v_2') \geq v_s'$: $A(\textib{v}') = b_1,\: p_1(\textib{v}') = -p_s(\textib{v}') = f_1(v_2'),\: p_2(\textib{v}') = 0$.
    \item Else, if $v_2' = f_2(v_1') \geq v_s'$: $A(\textib{v}') = b_2,\: p_2(\textib{v}') = -p_s(\textib{v}') = f_2(v_1'),\: p_1(\textib{v}') = 0$.
    \item Else, $A(\textib{v}') = s$ and $p_s(\textib{v}')=p_1(\textib{v}')=p_2(\textib{v}')=0$.
\end{enumerate}
\end{definition}

In other words, the mechanism receives  bids $\textib{v}'=(v_s',v_1',v_2')$ and calculates the values of the functions $f_1(v_2')$ and $f_2(v_1')$. Since the functions are compatible, at most one of the buyer's reports can be strictly greater than their corresponding payment function. If this is the case, the mechanism performs a trade with the seller whenever the gains in non-negative. Otherwise, one or both of the buyer's reports might be equal to their respective payment functions. In this case, performs a trade between the high value buyer and the seller, whenever the gains in non-negative.
Finally, if both buyers will not trade at their payment function's value, there is no trade, and the item remains in the seller's possession.

\begin{lemma}\label{lem:compmech}
For every pair of compatible functions $(f_1,f_2)$, the mechanism $M_{f_1,f_2}$ is a \nice mechanism and is associated with the pair $(f_1, f_2)$.
\end{lemma}

The proof of this lemma can be found in Appendix \ref{app:missing32}.

The mechanism described in Definition \ref{def:m} performs a trade whenever a trade increases the GFT, and it break ties in favor of higher GFT, and thus it %
maximizes GFT among all mechanisms that are associated with the pair of functions $(f_1,f_2)$:

\begin{lemma}\label{lem:mfmax}
    Let $M=(A,p)$ be a \nice mechanism, let $(f_1,f_2)$ be the pair of functions associated with $M$. Let $\textib{v}=(v_s,v_1,v_2)$ be some (truthful) bidding profile. Then:
    \begin{gather*}
    \textup{GFT}(M_{f_1,f_2},\textib{v}) \geq \textup{GFT}(M,\textib{v})
    \end{gather*}
\end{lemma}

We next present a sketch of the proof, for the full proof see Appendix \ref{app:missing32}. 
    If $v_b > f_b(v_{-b})$, then all mechanisms that are associated with $(f_1,f_2)$ must behave the same way, so the only difference in allocation can occur in cases of ties, such as when $v_1 = f_1(v_2)$ or $v_2 = f_2(v_1)$. In these cases, mechanism $M_{f_1,f_2}$ always chooses to trade with the buyer with the higher value - thus maximizing GFT. Additionally, when the seller's value is equal to the price, i.e. $v_s = f_1(v_2)$ or $v_s = f_2(v_1)$, the seller is ambivalent to the trade. When GFT is positive, the mechanism always trades in these cases, so for every $\textib{v}=(v_s,v_1,v_2)$ it has optional GFT over all mechanisms that are associated with $(f_1,f_2)$.

Taking the expectation over the valuation profiles \textib{v} we get the following corollary:

\begin{corollary}
      Let $M=(A,p)$ be a \nice mechanism, let $(f_1,f_2)$ be the pair of functions associated with  $M$. 
      For any joint prior $\textbf{V}$ over the  agents' valuations  it holds that:
     \begin{gather*}
    \textup{GFT}(M_{f_1,f_2},\textbf{V}) \geq \textup{GFT}(M,\textbf{V})
    \end{gather*}
\end{corollary}
With a slight abuse, throughout the paper we refer to $\textup{GFT}(M_{f_1,f_2},\textib{V})$ %
as the GFT of the pair $(f_1,f_2)$ on $\textib{V}$. The above corollary shows that $\textup{GFT}(M_{f_1,f_2},\textib{V})$ is the maximal GFT that is obtainable on $\textib{V}$ by any \nice mechanism for which $(f_1,f_2)$ are its pair of associated functions.  

The following summarizes the results of this section. %

\begin{restatable}{proposition}{mechpair}\label{prop:mech-pair}
    Every \nice mechanism $M$ is associated with a unique compatible pair of functions $(f_1, f_2)$. 
    Conversely, for every compatible pair of functions $(f_1, f_2)$ there exists a \nice mechanism $M_{f_1,f_2}$. 
    that is associated with this pair. 
    Furthermore, the \nice mechanism $M_{f_1,f_2}$ achieves the highest 
    gains from trade among all \nice mechanisms that are associated with the pair $(f_1, f_2)$.
\end{restatable}

\section{Impossibility of Learning for General Distributions}\label{sec:impos}

In this section we show that for unrestricted joint distributions (when values might be correlated), it is impossible to learn an approximately optimal \nice mechanism, even when the support is on $[0,1]^3$. To begin, we show that for every finite joint distribution that satisfies a mild condition (``generic support''), there is a \nice mechanism that achieves the maximum possible GFT (that is, there is a \nice mechanism that obtains the first-best GFT).
Subsequently, we leverage this finding to demonstrate scenarios where unrestricted joint distributions pose a challenge to the feasibility of learning approximately optimal mechanisms.

\subsection{First-best GFT by a \Nice Mechanism}

We proceed by introducing a family of distributions for which \nice mechanisms are able to perfectly maximize the GFT (obtain the first best). 
Specifically, for each distribution in that family, we present an optimal \nice mechanism that is specifically tailored to it, showcasing a form of extreme overfitting. The core concept behind this overfitting approach is as follows: consider a collection of triplets where each value is unique and does not repeat (``a generic collection''). 
When the support is generic, the value of the lower-value buyer completely reveals the values of the two other agents, so to extract maximum GFT the mechanism sets the trade price to be the value of the highest-value buyer. That mechanism is DSIC as the price does not depend on the reports of the trading agents. In fact, this yields a \nice mechanism that extracts the maximum GFT using the characterization of Proposition \ref{prop:mech-pair}.

Formally, to be able to extract all GFT it is enough that the distribution is supported on a generic set of vectors. 
A set $C = 
\{(v_1^1,v_2^1,\ldots,v_d^1), \ldots, (v_1^n,v_2^n,\ldots,v_d^n)\} \subseteq \mathbb{R}^d_{\geq 0}$ of $n$ vectors (of length $d$) is called \emph{generic}, if there are no two values that are identical (for all $i_1,i_2\in [d]$ and $j_1,j_2\in [n]$ it holds that $v_{i_1}^{j_1}\neq v_{i_2}^{j_2}$ unless $i_1=i_2$ and $j_1=j_2$).\footnote{For our proof to hold it is sufficient that a weaker condition holds: no value repeats more than once \emph{in any fixed coordinate $i\in[d]$}.}

For a given generic set we can construct a \nice mechanism that extracts all GFT for any input in that set, by setting the trade price to be the value of the highest-value buyer, deducing that price from the value of the other buyer. 

\begin{definition}\label{def:overfitm}
    Let $C = \{(v_s^1,v_1^1,v_2^1), \ldots, (v_s^n,v_1^n,v_2^n)\}$ be a generic set of vectors representing the support of a finite joint distribution $\textbf{V}$ for values in  the 1-seller 2-buyer setting. We define the  functions $f_1^C,f_2^C$ as follows:
    \begin{gather*}
        \forall b \in \{b_1,b_2\}: f_b^C(x) = \begin{cases}
            v_b^k & \exists k : (v_s^k,v_1^k, v_2^k)\in C \ s.t.\   
            x = v_{-b}^k\leq v_b^k\\
            \infty & \text{otherwise}
        \end{cases}
    \end{gather*}

    Let $M_{f_1^C,f_2^C}$ be the mechanism associate with the pair $(f_1^C,f_2^C)$,  
    as defined in Definition \ref{def:m}.
\end{definition}

Note that since $C$ is a generic set, these functions are well defined.

These functions capitalize on the knowledge encoded within $C$ to establish distinct prices for each triplet of values. They achieve this by consistently setting the price as the higher buyer's value while setting the other's to $\infty$. The price for each buyer is set using a condition dependent only on the report of the other buyer, as necessary for a DSIC mechanism.
This approach renders the higher buyer indifferent and thus prompts acceptance of the trade, whereas the lower buyer is either not offered the trade or is presented with an offer at $\infty$ (equivalent to non-offer). Additionally, the functions are compatible because for every $k\in [n]$, if the function $f_b^C(v_{-b}^k)$ is set to be some value that is not $\infty$, then either $f_{-b}^C(v_b^k) = \infty$ or $f_b^C(v_{-b}^k) = v_b^k = v_{-b}^k = f_{-b}^C(v_b^k)$ and both cases do no violate compatibility. 
Hence, the mechanism $M_{f_1^C,f_2^C}$ is straightforwardly a well-defined \nice mechanism,  
a fact that can also be directly deduced by Lemma \ref{lem:compmech}.

Finally, we need to demonstrate that the mechanism achieves the maximum GFT (first best):

\begin{lemma}\label{lem:fullgft}
    Let $C$ be a finite generic set, and let $\textbf{V}$ be a probability distribution over $C$. 
    Then the GFT of the \nice mechanism $M_{f_1^C,f_2^C}$ is the unconstrained optimum (first best): $\textup{GFT}(M_{f_1^C,f_2^C},\textbf{V}) = \FBV$. 
\end{lemma}

\begin{proof}

Recall that the optimal GFT is defined to be 
$$ \FBV = \mathbb{E}_{\textib{v}\sim \textbf{V}}\left[\max \{v_1-v_s,v_2-v_s,0\}\right]$$

To prove the claim we show that for every valuation $\textbf{v} = (v_s^k,v_1^k,v_2^k)\in C$, the GFT obtained by $M_{f_1^C,f_2^C}$ is exactly $\max \{v_1-v_s,v_2-v_s,0\}$.

We first consider the case that 
$\max\{v_1^k, v_2^k\} < v_s^k$. 
We prove the claim for $\max\{v_1^k, v_2^k\} =v_1^k$  (the case that $\max\{v_1^k, v_2^k\} =v_2^k$ is similar and is omitted).
In this case  $f_1^C(v_2^k)= v_1^k < v_s^k$ and $f_2^C(v_1^k)= \infty > v_2^k$, and there is no trade in $M_{f_1^C,f_2^C}$. The mechanism  obtains GFT of $0$, and that is optimal  as $\max \{v_1^k- v_s^k,v_2^k- v_s^k, 0\} =0$.

Conversely, if $\max\{v_1^k, v_2^k\} \geq v_s^k$, then $\textup{GFT}(M_{f_1^C,f_2^C},\textib{v}) = \max\{v_1^k, v_2^k\} - v_s^k$. 
We prove the claim for $\max\{v_1^k, v_2^k\} =v_1^k$ (the other case is similar). In that case 
$f_1^C(v_2^k)= v_1^k \geq v_s^k$ and $f_2^C(v_1^k)= \infty > v_2^k$, and there is trade in $M_{f_1^C,f_2^C}$ between buyer $b_1$ and the seller, obtaining GFT of $v_1^k - v_s^k\geq 0$.
This represents the maximum GFT achievable from the triplet $(v_s^k,v_1^k,v_2^k)$, leading to $\textup{GFT}(M_{f_1^C,f_2^C},\textib{v}) = v_1^k- v_s^k= \max \{v_1^k- v_s^k,v_2^k- v_s^k, 0\}$. 
\end{proof}

\subsection{Impossibility of Learning for General %
Distributions}

In this section we consider the problem of learning a \nice mechanism with GFT that is approximately optimal over all \nice  mechanisms, when the values are sampled from some unknown joint distribution (values may be  correlated), and we have sample access to the distribution. We prove that no learning algorithm is able to distinguish 
between the case that a \nice mechanism can obtain the optimal GFT (first best), and the case that every \nice mechanism can only obtain a significantly lower GFT.

Suppose that we have a learning algorithm that accepts some finite number $t$ of samples from an unknown distribution $\mathbf{V}$ and presumably learns a \nice mechanism that approximately maximizes GFT for $\mathbf{V}$.  We can certainly use such an algorithm to also tell whether for distribution $\mathbf{V}$ there is a large gap between the first-best GFT, and the GFT of the GFT-optimal \nice mechanism.

Recall that for distribution $\mathbf{V}$, we denote the optimum GFT (first best) by $\FBV$, the highest GFT of any \nice mechanism by $\OPTniceV$.  

\begin{lemma}\label{lem:sepnice} 
Assume that there exists a learning algorithm that accepts $t$ random samples from any unknown distribution $\mathbf{V}$ on $[0,1]^3$ and, with high probability, outputs a \nice mechanism with GFT that is within an additive $c$ from that of the best \nice mechanism for $\mathbf{V}$.  Then, there exists an algorithm that accepts $t+O(1/c^2)$ samples and, with high probability, distinguishes between the following two cases:
    \begin{enumerate}
        \item $\OPTniceV=\FBV$. 
        \item $\OPTniceV \le \FBV-6c$.
    \end{enumerate}   
\end{lemma}

\begin{proof}
    Denote by $S_t$ the sample that is comprised of the first $t$ samples from $\mathbf{V}$, by $S_c$ the sample that consists of the other $O(1/c^2)$ samples. Let $U(S_t)$ and $U(S_c)$ be the uniform distributions on these multi-sets of samples, respectively. Let $M$ be the \nice mechanism that is the result of running the assumed algorithm on $U(S_t)$. By assumption,  with high probability, the GFT of $M$ on $\mathbf{V}$ is within an additive $c$ from that of the best \nice mechanism for $\mathbf{V}$, that is: $|\textup{GFT}(M, V) - \OPTniceV|<c$.
    Denote $G_c=\textup{GFT}(M, U(S_c))$ the GFT obtained by running $M$ on $U(S_c)$, and $G^*_c = \textup{OPT}(U(S_c))$ the GFT of the first-best mechanism for $U(S_c)$ (i.e. average the value $max\{0, v_1-v_s, v_2-v_s\}$ over all $O(1/c^2)$ triplets).

    Notice that from the Chebyshev's inequality, with high probability, $|G^*_c - \FBV| = |\textup{OPT}(U(S_c))$ $ - \FBV|\leq c$. This holds since the GFT on $\mathbf{V}$ is the expected value of $max\{0, v_1-v_s, v_2-v_s\}$ over $\mathbf{V}$, while we took the average GFT of $O(1/c^2)$ samples which, with high probability, gives us a good estimate (specifically to within $c$) of this expected value. Similarly, with high probability, $|G_c-\textup{GFT}(M, \textbf{V})|\leq c$.
    Again, this holds since we estimated the expected value of $\textup{GFT}(M, \textbf{V})$  using the average of $O(1/c^2)$ samples. Combining with the fact that, 
    with high probability, $|\textup{GFT}(M, V) - \OPTniceV|<c$ (by our assumption on the learned mechanism), we derive that w.h.p. 
    $|G_c - \OPTniceV|<2c$.

    In case (1), as $\OPTniceV=\SBV=\FBV$, with high probability it holds that:
    \begin{gather*}
        G_c^* \leq \FBV + c = \OPTniceV + c \leq (G_c + 2c) +c=  G_c + 3c\\
        \Rightarrow G_c^* - G_c \le 3c
    \end{gather*}

    In case (2), as $\OPTniceV%
    <\FBV-6c$, then with high probability it holds that:
    \begin{gather*}
        G^*_c \geq \FBV - c >%
        \OPTniceV + 5c \geq (G_c-2c) +5c = G_c +3c \\
        \Rightarrow G_c^* - G_c > 3c
    \end{gather*}

    Thus by computing the difference $G_c^* - G_c$, with high probability we can distinguish between the two cases.    
\end{proof}

We will be applying this lemma to cases where
not only is $\OPTniceV$ smaller than $\FBV$, but even the GFT of
any Bayesian incentive-compatible and weakly budget-balanced mechanism, as studied 
by \cite{myerson1983efficient}, is
smaller.
So for explicitness let us denote
the GFT of the best mechanism from
the general class studied by \cite{myerson1983efficient} 
(the ``second best'' there) by  $\SBV$.  
Note that $\OPTniceV \le \SBV$
and whenever
$\FBV=\OPTniceV$ then also 
$\FBV=\SBV$.  
We now start with some uncorrelated distribution where the first-best $\FBV$ is strictly larger than the second-best 
$\SBV$ (and thus
certainly also larger than $\OPTniceV$) \cite{myerson1983efficient}. Specifically, $v_s, v_1$ and $v_2$ are independently sampled, each distributed uniformly on $[0,1]$.\footnote{\label{fn:gap}Actually, \cite{myerson1983efficient} does not explicitly discuss cases apart from bilateral trade. However, for 1-seller 2-buyer uniform distribution over $[0,1]^3$, consider the subcase when $v_s,v_1 \geq \frac{1}{2}$ and $v_2 < \frac{1}{2}$ (which happens with a constant probability). Since no trade with buyer 2 is possible, and since buyer $2$ value gives no useful information due to independence, we revert to a bilateral trade setting where \cite{myerson1983efficient} showed a  gap - and so there must also be a gap for uniform distribution over $[0,1]^3$.}  

Fix a parameter $T$ and consider a random choice of $T>>t$ triplets
$S=\{(v_s^i, v_1^i, v_2^i)\:|\:i \in [T]\}$ where
each of these $3T$ values are chosen uniformly at random in $[0,1]$, and look at the distribution on triplets $U_S$ that is 
uniform over the triplets in $S$. Since all values were chosen from a continuous distribution, clearly, with probability 1, no value appears more than once anywhere in $S$, and so we can apply Lemma \ref{lem:fullgft} to deduce that almost surely (over the choice of $S$) there exists a \nice mechanism that extracts the first-best gains from trade from $U_S$.

On the other hand, since we specifically started with a distribution for which there is a gap between the first best and second best (the uniform distribution over $[0,1]^3$), 
the GFT of any \nice mechanism is smaller than the first-best by some constant \cite{myerson1983efficient}. Assuming that we have an algorithm that learns, using $t$ samples, a \nice mechanism that (almost) maximizes gains from trade, then the previous lemma implies that we can distinguish
between the (correlated) distribution $U_S$, for a random $S$, 
and the uniform distribution over $([0,1]^3)^t$.  
But that, as we will show, is impossible, as these distributions are 
statistically close to each other.

\begin{lemma} \label{lem:cannot-tell}
Denote by $U_t$ the distribution on $t$ triplets chosen uniformly at random from $[0,1]^3$.
Let $S$ be chosen by taking a uniformly random sample of 
$T$ triplets from $[0,1]^3$, and denote by $U_S$ the distribution obtained
by choosing $t$ triplets from $S$ uniformly at random.  (Both $U_t$ and $U_S$
are distributions over $([0,1]^3)^t$.)  If $T > 6t^2$ then for any algorithm 
that accepts $t$ triplets, the
distribution of the output of the algorithm when fed a random input from $U_t$ is within
total variation distance of at most $1/6$ from the distribution of the output when fed a
random input from $U_S$.
\end{lemma}

\begin{proof}
The total variation distance between the output distributions is well known to be bounded by
the total variation distance between the input distributions, so it remains to bound this total variation distance from above by $1/6$ for the lemma to follow.

The distribution $U_t$ is clearly uniform over $([0,1]^3)^t$.  It remains to show that
$U_S$ has total variation distance of at most $1/6$ from uniform.  
Restricted to the subset of $([0,1]^3)^t$ where no real value appears
twice, $U_S$ is also uniform. It remains bound from above the probability
that $U_S$ places on the complement of this subset.
Note that with probability $1$ the set $S$ is generic, i.e. no real value appears twice anywhere in $S$.   Once $S$ has been chosen in a generic way the probability the $t$ samples
contain some 
repeated value 
can be easily estimated from above by $t^2/T < 1/6$. 
\end{proof}

Combining these two lemmata we get the result showing  that optimal simple mechanisms cannot be 
learned, and in fact it is even impossible to learn a general mechanism (Bayesian incentive-compatible and weakly budget-balanced in the sense of \cite{myerson1983efficient}) that
approximates the GFT of the best \nice mechanism.  

\nocortri*

\begin{proof}
    Assume that such a learning algorithm exists 
    for $c>0$ small enough so that $6c$ is smaller than the gap between the 
    first-best and the second best GFT of the uniform distribution over $[0,1]^3$, as ensured 
    by \cite{myerson1981optimal} (see Footnote \ref{fn:gap}).
    Lemma \ref{lem:sepnice} shows that the algorithm may
    be used to distinguish between a uniform distribution on any generic set for which, by Lemma
    \ref{lem:fullgft},
    $\OPTniceV=\FBV$, and the uniform distribution on 
    $[0,1]^3$ for which $\OPTniceV \le \SBV \le \FBV-6c$.   Since the
    distribution $U_S$ from the previous lemma is just the
    average over all possible sets $S$ of size $T$ (which are
    generic with probability 1) of the uniform distribution of $t$ samples from the set $S$, the algorithm (which accepts $t$ samples from $[0,1]^3$) also separates between $U_S$ and the uniform distribution on $([0,1]^3)^t$.  But this yields a contradiction as Lemma \ref{lem:cannot-tell} %
    states that no algorithm can achieve
    such a separation with high probability.
    
    In fact, even if the hypothetical learning algorithm is allowed to output any {\em general} mechanism that 
    provides GFT that are approximately at least as high as those of the best {\em \nice} mechanism then the exact same proof of Lemma \ref{lem:sepnice} 
    shows that this
    general mechanism can
    be used to distinguish between the two cases, since still the GFT of the learned 
    general mechanism in the case
    of the uniform distribution on $[0,1]^3$ would be separated from $\FBV$.
\end{proof}

\section{GFT-Optimal \Nice Mechanisms under Independence}%
\label{sec:optimal}
 
In the previous section we have shown that for
unrestricted joint distributions, it is impossible to learn an approximately GFT-optimal \nice mechanism. Given this impossibility result for arbitrary joint distributions, we move to focus on product distributions, where the valuations of the agents that are sampled independently.
Thus, we assume the value of seller $s$ is sampled from $V_s$, and the value of buyer $b_i\in \{b_1,b_2\}$ is sampled from $V_i$, so the triplet $(v_s,v_1,v_2)$ is sampled from the product distribution
$\textib{V}=V_s\times V_1 \times V_2$.
In this section we establish that under independence, the GFT of any \nice mechanism can also obtained by a \nice mechanism characterized by a compatible pair of associated functions which are monotone, and with a specific structure (being ``tight'').
These properties of the pair of functions will enable learning a \nice mechanism with GFT that is close to being GFT-optimal, as described in the following  sections.

Consider any \nice mechanism $M$ and its pair of associated functions $(f_1,f_2)$ which is compatible. 
These functions might not be monotone, as can be seen in the example that appears in Figure \ref{fig:randomexample}. We will show that when the valuations are independent, it is possible to modify the two functions and generate a pair of monotone functions $(f^*_1, f_2^*)$ that is still compatible, such that the GFT of $M_{f^*_1, f^*_2}$ is at least as high as the GFT of $M$. For this, we first need to introduce the concept of \emph{compatibility restriction}, which given one function deduce the constraints on the other function that are implied by the compatibility requirement.  

\begin{definition}[Compatibility Restriction]\label{def:comprest}
For buyer $b\in \{b_1,b_2\}$ and a function $f_{-b}: \mathbb{R}_{\geq 0} \rightarrow \mathbb{R}_{\geq 0}\cup \{\infty\}$, we define the \emph{compatibility restriction $r_b^{f_{-b}}(v_{-b})$ that $f_{-b}$ imposes on $f_{b}$} at $v_{-b}$ as follows:
    \begin{gather*}
        r_b^{f_{-b}}(v_{-b})=\begin{cases}
            \sup \{v_b | v_{-b} \geq f_{-b}(v_b)\}& \text{if } \exists v_b : v_{-b} \geq f_{-b}(v_b) \\
            0& \text{otherwise}
        \end{cases}
    \end{gather*}
We refer to $r_b^{f_{-b}}(\cdot)$ as \emph{the compatibility restriction function of $f_{-b}$ on $f_{b}$}.  
\end{definition}

An illustration of a compatibility restriction function can be seen in Figure \ref{fig:compres1}. The function $r_b^{f_{-b}}$ is clearly monotone non-decreasing,  as when $v_{-b}$ increases the supremum is taken over a superset.

\begin{figure}[htp]
\centering
\begin{minipage}{.45\textwidth}
  \centering
  \includegraphics[width=\linewidth]{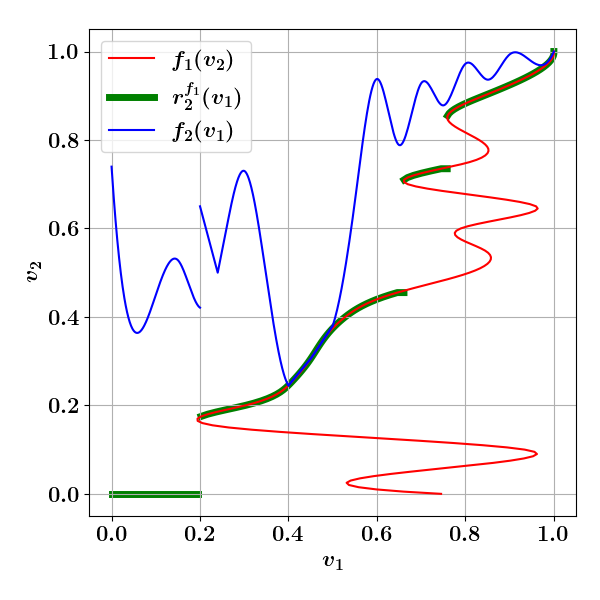} 
  \caption{An example of the Compatibility Restriction function $r_2^{f_1}$, the compatibility restriction of  $f_1$ on $f_2$.}
  \label{fig:compres1}
\end{minipage}%
\hfill
\begin{minipage}{.45\textwidth}
  \centering
  \includegraphics[width=\linewidth]{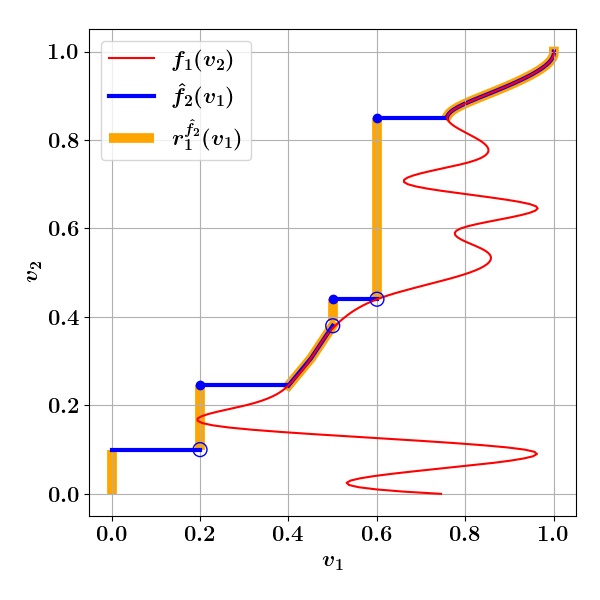} 
  \caption{The intermediate function $\hat{f}_2$, and the Compatibility Restriction function $r_1^{\hat{f}_2}$, the compatibility restriction of $\hat{f}_2$ on $f_1$.}
  \label{fig:f2hat}
\end{minipage}
\end{figure}

We next describe the three-step process that starts with a pair of compatible functions $(f_1,f_2)$, and ends with a  compatible and tight (see Definition \ref{def:tight}) pair of monotone functions $(f^*_1,f^*_2)$.
In each step the expected GFT of the mechanism that corresponds to the pair of functions (as per Definition \ref{def:m}) does not decrease. We now present how we do one of these step (the two other steps are similar). Consider any \nice mechanism $M$ and its pair of associated functions $(f_1,f_2)$ which is compatible. We want to turn $f_2$ into a monotone function $\hat{f}_2$ such that the pair $(f_1, \hat{f}_2)$
is compatible, and the GFT of $M_{f_1, \hat{f}_2}$ is at least as high as the GFT of $M$. To do that we fix $f_1$, and for every value $v_1$ define the value of $\hat{f}_2(v_1)$ to be the highest price over all prices that maximizes the GFT in the bilateral trade between buyer $b_2$ and seller $s$, under the constraint that the price is at least  $r_2^{f_{1}}(v_{1})$. The function $r_2^{f_{1}}$ is monotone non-decreasing and as a result the function  $\hat{f}_2$ is also monotone non-decreasing
(see proof in Appendix \ref{app:missing5}).   

In the three-step modification process, we first turn $f_2$ into a monotone function $\hat{f}_2$ such that the pair $(f_1, \hat{f}_2)$ is compatible, and the GFT of $M_{f_1, \hat{f}_2}$ is at least as high as the GFT of $M$. Figure \ref{fig:f2hat} shows this new function $\hat{f}_2$ and the compatibility restriction it imposes on $f_1$. Second, we turn $f_1$ into a monotone function $f^*_1$ such that the pair $(f^*_1, \hat{f}_2)$ is compatible, and the GFT of $M_{f^*_1, \hat{f}_2}$ is at least as high as the GFT of $M_{f_1, \hat{f}_2}$. The third and final step is to modify  the monotone function $\hat{f}_2$ and generate a monotone function $f^*_2$ such that the pair $(f^*_1, f^*_2)$ is compatible and does not lose any GFT, while being tight (see Definition \ref{def:tight}). 
\OLD{
\begin{definition}\label{def:tight}
    Let $(f_1,f_2)$ be a pair of compatible functions, and let $p_1,p_2 \in \mathbb{R}$. The pair $f_1,f_2$ is said to be \textup{tight after} $(p_1,p_2)$ if:
    $$\forall v_1 \geq p_1, \forall v_2 \geq p_2 : (v_1 \geq f_1(v_2) \vee v_2 \geq f_2(v_1))$$

    The functions are called \textup{tight} if there exists a point $(p_1,p_2)$ after which they are tight.
\end{definition}
In other words, this definition means that there is no point above $(p_1,p_2)$ where both $v_1 < f_1(v_2)$ and $v_2 < f_2(v_1)$. 
}
Recall that the pair $f_1,f_2$ is said to be tight after $(p_1,p_2)$ if there is no point above $(p_1,p_2)$ where both $v_1 < f_1(v_2)$ and $v_2 < f_2(v_1)$.
Figure \ref{fig:f1star_f2'} shows the functions $f_1^*,\hat{f}_2$, illustrating that they are not tight, and Figure \ref{fig:stars} shows the final functions $f_1^*,f_2^*$, that are tight.

\begin{figure}[htp]
\centering
\begin{minipage}{.45\textwidth}
  \centering
  \includegraphics[width=\linewidth]{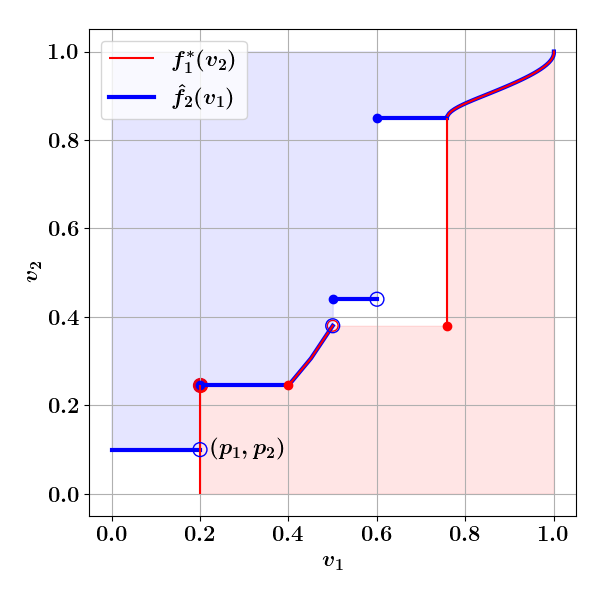} 
  \caption{Illustration of $f_1^*$ and $\hat{f}_2$. The lower-right shaded area (in red) represents points at which the mechanism will use the posted-price mechanism with price $f_1^*(v_2)$ on buyer $b_1$ and the seller. Similarly, the upper-left shaded area (in blue) represents 
  using price $\hat{f}_2(v_1)$ on buyer $b_2$ and the seller.}
  \label{fig:f1star_f2'}
\end{minipage}%
\hfill
\begin{minipage}{.45\textwidth}
  \centering
  \includegraphics[width=\linewidth]{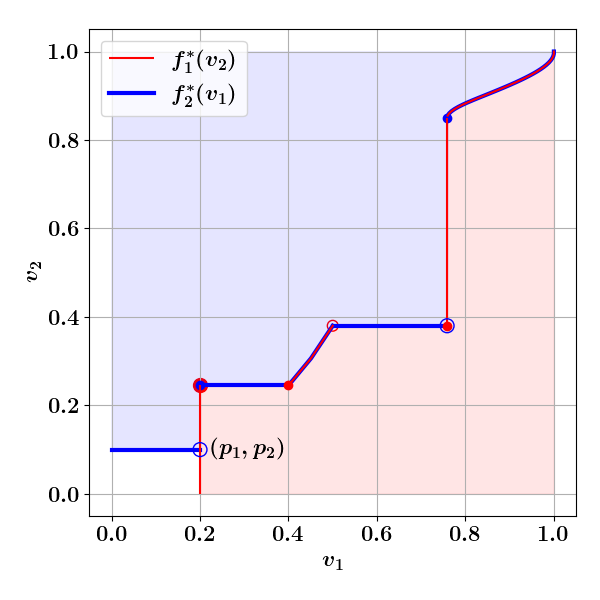} 
  \caption{Illustration of the final functions $f_1^*$ and $f_2^*$. Note that above and to the right of their first meeting point, there are no pairs of values strictly between the two curves (no blank spaces), i.e. they are tight.}
  \label{fig:stars}
\end{minipage}
\end{figure}

In summary, by applying these three modifications on a pair of functions $(f_1,f_2)$ that is associated with a \nice mechanism, we get another \nice mechanism $M_{f^*_1, f^*_2}$ that is associated with $(f^*_1,f^*_2)$ that are a compatible and tight pair of monotone functions. We have thus derived the main result of this section:

\bestmech*

\subsection{Proof of Theorem \ref{thm:bestmech}}
In this section we present the proof of  Theorem \ref{thm:bestmech}. 
As stated above, in each of the three steps in which we change one of the functions, we do so by setting the value at each point to be the best price over all prices that maximize GFT in the bilateral trade scenario, under the constraint that the price is not lower than the appropriate compatibility restriction. Formally, we define the best price under a restriction:

\begin{definition}[Restricted Best Price]\label{def:resbp}
        Let $V_s,V_b$ be bounded distributions. We define $p^*(r, V_s\times V_b)$, %
         the \textup{Restricted Best Price} for distribution $V_s\times V_b$ %
         given restriction $r\in\mathbb{R}$,
         to be:
    \begin{gather*}
        p^*(r, V_s\times V_b) = \sup\left\{\argsup\limits_{p \geq r} \textup{GFT}(p, V_s\times V_b)\right\}
    \end{gather*}
    
    When $r=0$ there is no restriction on price and we call $p^*(0, V_s\times V_b)$ the \emph{best price} for distribution $V_s\times V_b$.
\end{definition}

In other words, for bilateral trade between seller $s$ and buyer $b$, the restricted best price is essentially the highest price that is at least $r$ that obtained the maximal GFT for distribution $V_s\times V_b$ (over all prices that are at least $r$). 

Note that the definition is using  \textit{sup} and \textit{argsup}, so as defined, it is not clear that the defined price itself maximizes the GFT.  Yet, we prove that it actually does, showing we can substitute the \textit{sup} and \textit{argsup} with \textit{max} and \textit{argmax}, respectively. 
Notably, this isn't a straightforward result and typically it would necessitate the GFT function to be continuous — a condition that isn't met in this case. However, upper semi-continuity suffices for our specific context, and GFT is indeed upper semi-continuous. Additionally, the restricted best price is also monotone non-decreasing, which in combination with the fact that any compatibility restriction function in also monotone non-decreasing will aid us in modifying the payment functions to be monotone. Detailed proofs of these properties are available in Appendix \ref{app:missing5}.

We are now ready to define the operator that will be iteratively executed three times, each time on a different agent.
The operator takes as input the function $f_{-b}$, associated with buyer $-b$, and the distributions for the seller $s$ and buyer $b$, and outputs a new function $\tilde{f}_b$ for $b$ that gives the best price for the product distribution $V_s\times V_b$ under the constraint imposed by $f_{-b}$.

\begin{definition}\label{def:f'}
   Fix some distributions $V_S, V_1,V_2$, and let $(f_1,f_2)$ be a pair of compatible functions. Define the function $\tilde{f}_b = g(f_{-b},V_s\times V_b)$ as the restricted best-price function on $V_s\times V_b$, restricted by the function $f_{-b}$. That is, for value $v_{-b}$ the function is defined as follows:
    \begin{gather*}
        \tilde{f}_b(v_{-b}) = p^*(r_b^{f_{-b}}(v_{-b}), V_s\times V_b)
    \end{gather*}

    The three modification steps we do are defined as follows:
    \begin{align*}
        \hat{f}_2 &= g(f_1,V_s\times V_2) & f_1^* &=  g(\hat{f}_2,V_s\times V_1) & f_2^* &= g(f_1^*,V_s\times V_2)
    \end{align*}
\end{definition}

Note that even though we notate them simply as $\hat{f}_2,f_1^*,f_2^*$, the function $\hat{f}_2$ actually depends on $f_1$, the function $f_1^*$ depends on $\hat{f}_2$, and the function $f_2^*$ depends on $f_1^*$. Also, they all depend on the distributions $V_s,V_1$ and $V_2$. An illustration of the function $\hat{f}_2$ (which is the result of the first modification step) is presented in Figure \ref{fig:f2hat}.

Crucially, this operator makes the new function monotone, while not violating compatibility:

\begin{lemma}\label{lem:starcompmono}
     Fix some distributions $V_S, V_1,V_2$, and let $(f_1,f_2)$ be a pair of compatible functions. The functions $\tilde{f}_1 =g(f_1,V_s\times V_2)$ and $\tilde{f}_2 = g(f_2,V_s\times V_1)$ are monotone non-decreasing; additionally, the pair $(\tilde{f}_1,{f}_2)$ is compatible, and so is the pair $({f}_1,\tilde{f}_2)$.
\end{lemma}

Additionally, modifying the functions according to the compatibility restrictions does not cause the new mechanism to lose any GFT, because we changed only one of the functions, and we only (weakly) improved it at each point:

\begin{lemma}\label{lem:moregft}
     Let $\textbf{V}=V_s\times V_1 \times V_2$ be a product distribution, and let $(f_1,f_2)$ be a pair of compatible functions. Denote $\tilde{f}_1 = g(f_2,V_s\times V_1)$ and $\tilde{f}_2 = g(f_1,V_s\times V_2)$. Then $\textup{GFT}(M_{\tilde{f}_1,f_{2}},\textbf{V}) \geq \textup{GFT}(M_{f_1,f_2},\textbf{V})$ and $\textup{GFT}(M_{{f}_1,\tilde{f}_{2}},\textbf{V}) \geq \textup{GFT}(M_{f_1,f_2},\textbf{V})$. %
\end{lemma}

Detailed proofs of these lemmata can be found in Appendix \ref{app:missing5}.

We proceed by explaining the reason behind iterating over the operator three times, alternating between the two agents. After two iterations  we have a pair of compatible functions are monotone non-decreasing. This is sufficient for learning a GFT-optimal mechanism in polynomial time, but by taking an extra iteration we can make learning easier. 

To understand this statement, we first point out that there is an `empty' rectangle (no trade) between $(0,0)$ and some point $(p_1,p_2)$, this is illustrated in Figure \ref{fig:f1star_f2'}. The values $p_1,p_2$ are the \textbf{unrestricted} best prices, and below them it is preferable for no trade to occur. This `empty' rectangle $(0,0)-(p_1,p_2)$ appears in many different distributions, for two examples, see Figure \ref{fig:uhalf} and Figure \ref{fig:complexexample}.

In Figure \ref{fig:f1star_f2'}, the colored areas represent the points where a trade occurs (given that the seller agrees) - red for trades $s\rightarrow b_1$ and blue for $s\rightarrow b_2$. Notice that above the point $(p_1,p_2)$, there are some `blank' rectangles. They are a result of defining $\hat{f}_2$ on the restrictions from $f_1$, which might be arbitrary. 
To close these gaps, we perform a third iteration, defining $f_2^*$ using the restriction from $f_1^*$ (the restriction $r_2^{f_1^*}(v_1)$ at each $v_1$). Intuitively, this means that wherever possible we `lower' $\hat{f}_2$'s value to the previous restricted best price, because now that price is allowed. An illustration of the result of this step can be seen in Figure \ref{fig:stars}. 

Again, this new function pair $(f_1^*, f_2^*)$ is compatible and both of the functions are monotone. However, now we also claim that they are tight - meaning they describe the same curve after the point $(p_1,p_2)$:

\begin{lemma}\label{lem:tight}
   Let $\textbf{V}=V_s\times V_1 \times V_2$ be a product distribution, and let $(f_1,f_2)$ be a pair of compatible functions. Let $p_1,p_2 \in \mathbb{R}$ be the best prices of $V_s\times V_1$ and $V_s\times V_2$ respectively, and let $f_1^*,f_2^*$ be defined as in Definition \ref{def:f'}. Then the pair $(f_1^*,f_2^*)$ is tight after $(p_1,p_2).$
\end{lemma}

\begin{proof} %
Recall that by the definition of a tight pair we need to show that there do {\em not} exist any values $v_1 \geq p_1, v_2 \geq p_2$ such that $v_1 < f_1^*(v_2)$ and $v_2 < f_2^*(v_1)$.

First of all, we claim that $\forall v_1: \hat{f}_2(v_1) \geq f_2^*(v_1)$. This is because for $f_2^*$ we used the restrictions from $f_1^*$, which is itself restricted from $\hat{f}_2$. Therefore, $f_2^*$ could use the values of $\hat{f}_2$, or improve by lowering the price. Increasing it will never occur, because the restricted best price always takes the maximum among best prices - and so no improvement would be made by increasing the price.

Now we fix $v_1 \geq p_1, v_2 \geq p_2$. If $v_1 \geq f_1^*(v_2)$ - we are done. Otherwise $v_1 < f_1^*(v_2)$. Recall that $\hat{f}_2$ is monotone, and that $\forall v_1: \hat{f}_2(v_1) \geq f_2^*(v_1)$. Additionally, $f_1^* \geq r_1^{\hat{f}_2}(v_2) = \sup \{v_1 | v_2 \geq \hat{f}_2(v_1)\} > v_1$, and so $\exists v_1' > v_1 : v_2 \geq \hat{f}_2(v_1')$. Therefore also $v_2 \geq f_2^*(v_1)$.
\end{proof}

From this lemma we learn that the pair $(f_1^*,f_2^*)$ is indeed tight after $(p_1,p_2)$, and above that point there will always be a trade when the seller's value is low enough - as illustrated in Figure \ref{fig:stars}.

In summary, we prove the main theorem of the section:

\begin{proof}[Proof of Theorem \ref{thm:bestmech}]
    Let $M$ be a \nice mechanism, let $(f_1,f_2)$ be the associated pair of functions of $M$, and let $\hat{f}_2, f_1^*$ and $f_2^*$ be defined as in Definition \ref{def:f'} for $(f_1,f_2)$ and the product distribution $\textib{V}$. 
    By using Lemma \ref{lem:starcompmono} three times, on the pairs of function $(f_1,\hat{f}_2)$, $(f^*_1,\hat{f}_2)$ and $(f^*_1, f^*_2)$, we derive that the pair of functions $(f_1^*, f_2^*)$ is compatible, each function is monotone non-decreasing and by Lemma \ref{lem:tight} they are tight. Thus, by Lemma \ref{lem:compmech}, $M_{f_1^*,f_2^*}$ is well-defined and \nice. Finally, by Lemma \ref{lem:moregft} $\textup{GFT}(M_{f_1^*,f_2^*},\textib{V}) \geq \textup{GFT}(M_{f_1,f_2},\textib{V})$.
\end{proof}

\section{Computing a GFT-optimal \Nice Mechanism under Independence}\label{sec:alg}

Before presenting our methodology for learning GFT-optimal mechanisms from samples, we introduce an algorithm designed to find a \nice mechanism that maximizes GFT on a product distribution with finite support,
when the product distribution is given to the algorithm as input. This will be useful in the next section, in which we construct a mechanism when we do not know the product distribution; rather, we have sample access to it. 
We show that when we only have sample access to a product distribution, 
we can learn a good mechanism by running the algorithm we present in this section on the empirical distribution generated from enough samples. 

In this section we consider the problem of finding a GFT-optimal mechanism for a given finite product distribution $\textib{V}=V_s\times V_1\times V_2$. %
First, to aid our description and computation of the algorithm, we construct an ordered set $S= \{s_1 < \ldots <s_m\}$ of size $m$, %
which is a union of the finite supports of the distributions $V_s,V_1,V_2$, so that the realized value of every agent is always in $S$. Clearly, it is sufficient to define the value of the functions $f_1^*$ and $f_2^*$ on on values in $S$. Moreover, we claim that to maximize the GFT, it is sufficient to consider only functions that output values in $S$. 
This is so because the mechanisms we consider break ties in favor of trade and in favor of the higher value.
Therefore, rounding a price that is not in $S$ to the closest value in $S$ would result in the exact same trade in all cases. This fact also allows us to assume that the (unrestricted) best prices $p_1,p_2$ are in $S$.

We leverage the result presented in Theorem \ref{thm:bestmech}%
, and use a dynamic programming approach to efficiently compute such a mechanism. Additionally, observe that as we only consider functions $f_1,f_2: S\rightarrow S$, since $S$ is finite we consider only a finite number of mechanisms. Therefore the maximum GFT is attainable and there exists a GFT-optimal mechanism.
By this and Theorem \ref{thm:bestmech}, we know there is a GFT-optimal mechanism with an associated compatible pair $(f_1^*,f_2^*)$ of monotone functions that have the following properties: 1) There is a point $(p_1,p_2)$ which represents the unrestricted best prices for each, and up to that point the functions are equal to those values. 2) After $(p_1,p_2)$ the functions are tight.

Before formalizing our result, we establish some useful notations. %
We use $V^{\geq r}$ to denote the distribution of $V$ conditioned on the value being at least $r$, i.e. $Pr[V^{\geq r} = v] = Pr[V = v | V \geq r]$. Similarly, we use the notation $V^{<r}$ for the conditional distribution defined by $Pr[V^{< r} = v] = Pr[V = v | V < r]$.

In the first step of the algorithm we find the best prices $p_1,p_2$ for distributions $V_s\times V_1$ and $V_s\times V_2$, respectively, by iterating over all elements in $S$ (for simplicity, we denote them as $p_1,p_2$ and not as indices in $S$). Then we construct a matrix $G$ that represents values in $S\times S$ in the range  from $(p_1,p_2)$ to $(s_m,s_m)$, where each element is defined as follows:

$$
    G[i,j] = \max_{\substack{\text{Compatible }(f_1,f_2)\text{ s.t.}\\(f_1(s_j) = s_i) \vee (f_2(s_i) = s_j)}}\left\{\textup{GFT}\left(M_{f_1,f_2}, \left(V_s \times V_1^{\geq s_i}  \times V_2^{\geq s_j}\right)\right)\right\} \cdot Pr[V_1 \geq s_i] \cdot Pr[V_2 \geq s_j]
$$

In other words, each element $G[i,j]$ contains the maximum GFT an optimal \nice mechanism can achieve on the product distribution $V_s \times V_1^{\geq s_i}  \times V_2^{\geq s_j}$, times the probability that $v_1 \geq s_i$ and $v_2 \geq s_j$, and conditioned on the fact that the functions $f_1,f_2$ begin at, or pass through, $s_i,s_j$.
For the functions to begin at, or pass through, $s_i,s_j$ they must meet the following condition: $(f_1(s_j) = s_i) \vee (f_2(s_i) = s_j)$.

Since at this step we only need to find the functions $f_1,f_2$ from $(p_1,p_2)$ on-wards and we know that they pass through that point, this condition is met at $(p_1,p_2)$ and so if the matrix is filled in correctly we will be able to use it to construct the functions $f_1,f_2$ from $(p_1,p_2)$ up to $(s_m,s_m)$. Figure \ref{fig:parts} illustrates the two parts of the functions that we compute in the algorithm (up to $(p_1,p_2)$, and from there onward). The  definition of $G[i,j]$ requires considering all compatible pairs of functions that begin at $s_i,s_j$, and there are exponentially many such compatible pairs of functions. However, in the next subsection we show how the matrix can nevertheless be filled in polynomial time.

When the matrix is full, we can traverse it by following the values that were chosen at each point, thus constructing the tight part of the compatible and tight pair of functions $(f_1^*,f_2^*)$. Finally, by combining this with the part that sets the best prices $p_1,p_2$ we get the pair of functions $(f_1^*,f_2^*)$ that are associated with a GFT-optimal mechanism for the given product distribution.

\begin{figure}[htp]
\centering
\begin{minipage}{.45\textwidth}
  \centering
  \includegraphics[width=\linewidth]{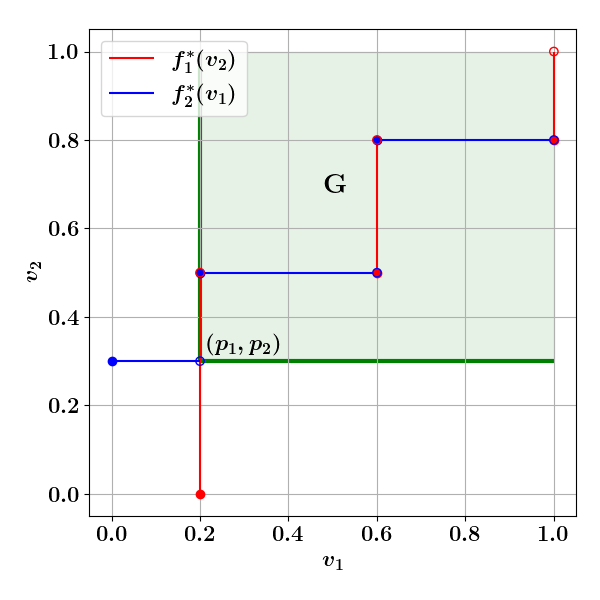} 
  \caption{Algorithm \ref{alg:maxgft} finds the values of $f_1^*,f_2^*$ in two parts. The best prices $p_1,p_2$ are found separately, and the matrix G is computed for the value above $(p_1,p_2)$, seen here in the green rectangle.}
  \label{fig:parts}
\end{minipage}%
\hfill
\begin{minipage}{.45\textwidth}
  \centering
  \includegraphics[width=\linewidth]{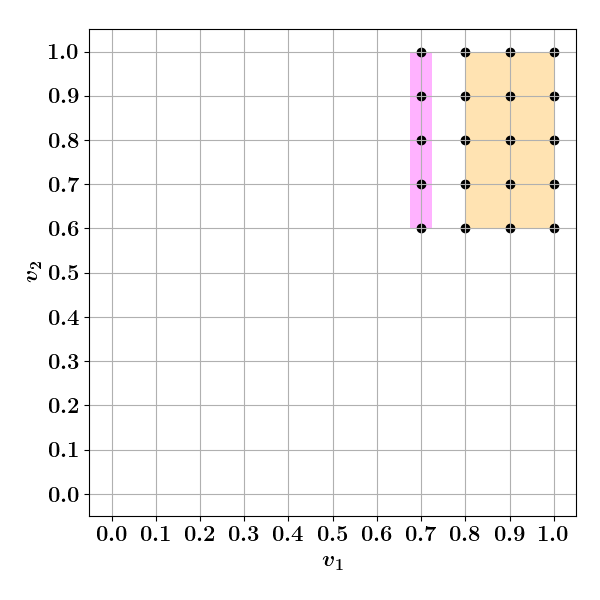} 
  \caption{Option (1) of the simple computation: By setting $f_2(0.7) = 0.6$ we gain GFT from the magenta column and $V_s$, and we add the optimal GFT that can be gained from the orange rectangle.}
  \label{fig:trunc}
\end{minipage}
\end{figure}

In summary, we present the main result of this section:

\algmax*

\subsection{Description of the Algorithm}

As stated previously, there is an exponential number of compatible and tight pairs $(f_1,f_2)$ of monotone functions $f_1,f_2: S\rightarrow S$ that begin at $s_i,s_j$. However we can use the condition that either $f_1(s_j)=s_i$ or $f_2(s_i)=s_j$, and that by Theorem \ref{thm:bestmech} there is a pair of tight and monotone functions that maximizes GFT, in order to reduce the number of functions that we must consider to only polynomially many. We first prove the following lemma regarding ties of these functions:

\begin{lemma}\label{lem:noties}
    Let $\textbf{V} = V_s \times V_1 \times V_2$ be a finite product distribution supported on the ordered set $S$, and let $p_1, p_2$ be the best prices for $V_s\times V_1$ and $V_s\times V_2$, respectively. Denote $m = |S|$, and let $ s_m \geq s_i \geq p_1, s_m \geq s_j \geq p_2$. 
    There exists a GFT-optimal mechanism $M$ s.t. the only point where $(f_1(s_j) = s_i) \wedge (f_2(s_i) = s_j)$ 
    is where $i = j = m$.
\end{lemma}

\begin{proof}
    Towards a contradiction, assume there is another such point $(s_i,s_j)$ and w.l.o.g. $i < m$. Since $(f_1(s_j) = s_i) \wedge (f_2(s_i) = s_j)$ we know the mechanism must break the tie when the players report $(s_i,s_j)$. Since the mechanism is optimal, we know the mechanism breaks the tie in favor of the larger value or arbitrarily if thy are equal. W.l.o.g. assume that the mechanism breaks the tie in favor of buyer 2, i.e. $s_j \geq s_i$. However, we can change $f_1(s_j) = s_{i+1}$ and the new mechanism will be equivalent, because for buyer 1 to win he must bid at least $s_{i+1}$. Therefore, there exists such an optimal mechanism $M$ where the only tie can be $i = j = m$. 
\end{proof}

Consider the case that $i = j = m$ does not hold. In that case, if $f_1(s_j)=s_i$ then $f_2(s_i) > s_j$, because from the compatibility property $f_2(s_i) \not < s_j$, and from Lemma \ref{lem:noties} it holds that $f_2(s_i) \neq s_j$. 
Similarly, if $f_2(s_i)=s_j$ then $f_1(s_j) > s_i$. Therefore we can consider $G[i,j]$ as the maximum of only two options:

\begin{align*}
    & G[i, j] = \max \Bigl\{\\
    & Pr[v_1 = s_i] \cdot Pr[v_2 \geq s_j] \cdot Pr[s_j \geq v_s] \cdot (\mathbb{E}[v_2\ |\  v_2 \geq s_j] - \mathbb{E}[v_s\ |\  s_j \geq v_s]) + G[i+1, j] \\
    & Pr[v_1 \geq s_i] \cdot Pr[v_2 = s_j] \cdot Pr[s_i \geq v_s] \cdot (\mathbb{E}[v_1\ |\  v_1 \geq s_i] - \mathbb{E}[v_s\ |\  s_i \geq v_s]) + G[i, j+1] \Bigr\}
\end{align*}

The first option handles the case that 
$f_2(s_i) = s_j$. In this case, we count the GFT from setting a price of $s_j$ for $s\rightarrow b_2$, which is $ Pr[v_2 \geq s_j] \cdot Pr[s_j \geq v_s] \cdot (\mathbb{E}[v_2\ |\  v_2 \geq s_j] - \mathbb{E}[v_s\ |\  s_j \geq v_s])$. We need to multiply that by the probability that $v_1 = s_i$. Finally, since we set $f_2(s_i) = s_j$ we know that $f_1(s_j) > s_i$ so we continue to $G[i+1,j]$ which considers the maximum GFT that can be gained on $V_s \times V_1^{\geq s_{i+1}}  \times V_2^{\geq s_j}$, times the probability that $v_1 \geq s_{i+1}$ and $v_2 \geq s_j$ and conditioned on the fact that the functions $f_1,f_2$ begin at $s_{i+1},s_j$. We can consider only the GFT from this point onwards because $f_1(s_j) > s_i$ and $f_2(s_i) = s_j$. The second option is symmetric.

To illustrate this, we present an example in Figure \ref{fig:trunc}. 
In this example $S$ includes all multiples of $0.1$. %
Assume that the goal is to maximize the GFT on the rectangle above $(0.7,0.6)$, i.e. $s_{i} = 0.7$ and $s_{j} = 0.6$. 
In the illustration, the magenta column represents the expected GFT from $\textup{GFT}(0.6, V_s\times V_2)$ when $v_1 = 0.7$. The orange rectangle represents the maximum possible GFT from the point $(0.8,0.6)$. Specifically, this is option (1) that appears above. %

As we now have a recursive definition of $G$, we can compute it using a dynamic program: Algorithm \ref{alg:maxgft} presents a pseudo-code description of the algorithm. %

The algorithm is separated into a few steps. First, we initialize the arrays we plan on using. The $P$ arrays are only needed for the buyers, and at element $i$ they contain the probability of a buyer having a specific value $s_i$. The $C$ and $E$ arrays are created for all three agents. The $C$ arrays represent the probability of a value being at least $s_i$ (or at most that value for the seller). The $E$ arrays contain at element $i$ the expected value conditioned on the fact that the random variable is at least $s_i$ (or at most for the seller). Initializing these arrays at the beginning will be helpful later in the algorithm.

Step 1 finds the best prices for $V_s\times V_1$ and $V_s\times V_2$ by iterating over all possible values in $S$ (which is sufficient as we have stated previously). Step 2 initializes the matrix G and sets the values to 0 where at least one of the coordinates is $m+1$ (this is outside the support, and is defined so for convenience). It also initializes $G_p$ which is a `pointer' matrix, which we will use in Step 4 when constructing $f_1^*,f_2^*$.

\begin{algorithm}[Computing a GFT-Optimal Mechanism for a Finite Product Distribution]\label{alg:maxgft}
    \textbf{Input:} Finite  distributions $V_s, V_1, V_2$ \\
    \noindent \textbf{Output:} A pair of functions $(f_1^*,f_2^*)$ such that $M_{f_1^*,f_2^*}$ GFT-optimal for $V_s\times V_1\times V_2$\\
    \\
    \textbf{Initialization:}
    
    \quad $S = supp(V_s) \cup supp(V_1) \cup supp(V_2) \cup \{0,1\}$, \texttt{sort $S$ in increasing order}, $m = |S|$

    \quad \texttt{Initialize $P_1, P_2, C_s, C_1, C_2, E_s, E_1, E_2, f_1^*,f_2^*$ as arrays of length $m$:}
    
    \quad $\forall b\in \{1,2\},\forall i \in [m]: P_b[i] = Pr[V_b = s_i],\  C_b[i] = Pr[V_b \geq s_i],\  E_b[i] = \mathbb{E}[V_b | V_b \geq s_i]$
    
    \quad $\forall i \in [m]: C_s[i] = Pr[V_s \leq s_i],\  E_s[i] = \mathbb{E}[V_s | V_s \leq s_i]$\\ 

    \noindent \textbf{Step 1:} Get Best Prices

    \quad $p_1 = \max_{i\in S}\{(E_1[i] - E_s[i])\cdot C_1[i] \cdot C_s[i]\}$ \ \ \  (Best price for $V_s\times V_1$)

    \quad $p_2 = \max_{i\in S}\{(E_2[i] - E_s[i])\cdot C_2[i] \cdot C_s[i]\}$ \ \ \  (Best price for $V_s\times V_2$)\\
    
    \noindent \textbf{Step 2:} Initialize the matrix $G$

    \quad \texttt{Initialize $G,G_p$ as a $(m+1)\times (m+1)$ matrices}

    \quad $\forall i,j \in [m+1]: G[i, m+1]=0, G[m+1,j] =0$\\

    \noindent \textbf{Step 3:} Fill $G$

    \quad \texttt{For {($i = m; S[i] \geq p_1; i--$)}:}
    
    \quad \quad \texttt{For {($j = m; S[j] \geq p_2; j--$)}:}

    \quad \quad \quad $G_1 = P_1[i] \cdot C_2[j] \cdot C_s[j] \cdot \Bigl(E_2[j] - E_s[j]\Bigr) + G[i+1,j]$
    
    \quad \quad \quad $G_2 = C_1[i] \cdot P_2[j] \cdot C_s[i] \cdot \Bigl(E_1[i] - E_s[i]\Bigr) + G[i,j+1]$

    \quad \quad \quad If $G_1 \geq G_2$:
    
    \quad \quad \quad \quad $G_p[i,j] =\ \rightarrow$

    \quad \quad \quad Else:

    \quad \quad \quad \quad $G_p[i,j] =\ \uparrow$
    
    \quad \quad \quad $G[i,j] = \max\{G_1, G_2\}$\\

    \noindent \textbf{Step 4:} Construct $f_1^*, f_2^*$

    \quad $\forall i \in [p_1], \forall j \in [p_2]: f_1^*[j] = p_1, f_2^*[i] = p_2$

    \quad $i = p_1, j = p_2$

    \quad \texttt{While $i \leq m$ or $j \leq m$ do:}

    \quad \quad \texttt{if $G[i,j] == \rightarrow$:}

    \quad \quad \quad $f_2^*[i] = j$
    
    \quad \quad \quad $i++$

    \quad \quad \texttt{else:}

    \quad \quad \quad $f_1^*[j] = i$
    
    \quad \quad \quad $j++$\\

    \noindent \textbf{Step 5:} Return $f_1^*,f_2^*$

\end{algorithm}

In Step 3 we fill the matrix according to the rule we described above. At $G[i,j]$ we consider both options: trading $s\rightarrow b_1$ at $s_i$ plus the GFT in $G[i,j+1]$ and trading at $s\rightarrow b_2$ at $s_j$ plus $G[i+1,j$. We set the value of $G[i,j]$ to be the maximum of these two options, and we set $G_p[i,j]$ to point to the direction that we used.

Step 4 constructs the payment functions from $G$. First, we set $\forall p_2\geq j \geq 0: f_1(j) = p_1$ and $\forall p_2\geq i \geq 0: f_2(i) = p_2$. Then we traverse the matrix from $G_p[p_1,p_2]$ to $G_p[m,m]$ by following the pointers. We set them to remember the values that we used at each step, so $G_p[i,j]$ points in the direction that we need to proceed.

It is clear that Step 3 fills the matrix at $G[i,j]$ according to the simplification of the computation we presented above. Additionally, it sets the base cases appropriately in Step 2, and constructs $f_1^*,f_2^*$ by following the maximum path we saved while filling $G$ - therefore Algorithm \ref{alg:maxgft} is correct.

\subsection{Proof of Theorem \ref{thm:algmax}}
We are now ready to prove the main theorem of this section:

\begin{proof}[Proof of Theorem \ref{thm:algmax}]
    Consider the mechanism $M_{f_1^*,f_2^*}$ which is the output of Algorithm \ref{alg:maxgft}. From the definitions of best prices we know that up to $p_1,p_2$ this mechanism achieves maximum GFT from trading at the best prices. This accounts for the GFT from $V_s \times V_1^{<p_1}  \times V_2$ and $V_s \times V_1 \times V_2^{< p_2}$. Also, by Lemma \ref{lem:tight} we know that there exists a \nice mechanism $M$ with an associated pair $(f_1^*,f_2^*)$ that are tight after $(p_1,p_2)$ and that these are unrestricted best prices, so either $f_1^*(p_2) = p_1$ or $f_2^*(p_1) = p_2$.
    
    As described above, we know that the maximum GFT that can be achieved on $V_s \times V_1^{\geq s_i} \times V_2^{\geq s_j}$ given that either $f_1(p_2) = p_1$ or $f_2(p_1) = p_2$ is equal to $G[i,j]$. Also, Algorithm \ref{alg:maxgft} constructs the functions so they achieve this exact GFT. So overall, the mechanism $M_{f_1,f_2}$ that uses the resulting functions from Algorithm \ref{alg:maxgft} achieves the maximum GFT possible for a \nice mechanism over all $V_s \times V_1 \times V_2$. Additionally, it is clear that Algorithm \ref{alg:maxgft} runs in polynomial time.
\end{proof}

\section{Learning to Maximize GFT under Independence}\label{sec:learn}

In this section we consider learning GFT-optimal mechanisms for a product of distributions that are supported on $[0,1]$ (when the support is not necessarily finite), while only having sample access to the distributions. Given parameters $\delta, \varepsilon>0$ we aim to find, with probability at least $1-\delta$, a \nice mechanism with GFT that is at most an additive $\varepsilon$ lower than the GFT-optimal mechanism tailored to the product distribution (which is unknown).
The learning procedure will run the algorithm presented in Section \ref{sec:alg} on the empirical product distribution that is derived from a polynomial (in $1/\varepsilon$ and $\log 1/\delta$) number of samples, finding the GFT-optimal mechanism for the empirical distribution. We show that the resulting mechanism indeed satisfies the requirements. 

The main tool we use is the $\varepsilon$-sample:

\begin{definition}[Def 14.6 from \cite{mitzenmacher2017probability}]
     Let $(X, \mathcal{R})$ be a range space and let $\mathcal{D}$ be a probability distribution on $X$. A set $S\subseteq X$ is an $\varepsilon$-sample w.r.t $\mathcal{D}$ if for all sets $R\in \mathcal{R}$,
     \begin{gather*}
         \left| Pr_\mathcal{D}[R] - \frac{|S\cap R|}{|S|} \right| \leq \varepsilon
     \end{gather*}
\end{definition}

An $\varepsilon$-sample is one in which the empirical distribution over the samples is almost the same as the underline distribution.
Intuitively, sampling enough points from a distribution creates a sample that represents the original distribution pretty well. In the case of an $\varepsilon$-sample, the original distribution is sampled enough times so that, with high probability, for every range in the probability space its weight in the sample is $\varepsilon$-close to its true weight in the distribution.

For some bounded distributions $V_s, V_1, V_2$ we denote their respective $\varepsilon$-samples by  $S_s, S_1, S_2$. We also denote  the uniform distribution over the multiset $S_i$ by $U_i$, i.e. $U_i = Uniform(S_i)$. We use $\textbf{U}=U_s \times U_1 \times U_2$ to denote the product of these uniform distributions, so sampling from $\textbf{U}$ returns a triplet of values $(u_s,u_1,u_2)$ which were each sampled independently from $U_s,U_1,U_2$, respectively. 

We now claim that a mechanism's GFT on the underlying product distribution that was sampled is close to its GFT on the $\varepsilon$-samples:

\begin{lemma}\label{lem:gftbound}
    Let $\textbf{V}=V_s\times V_1 \times V_2$ be a product distribution 
    over %
    $[0,1]^3$, and let $S_s, S_1, S_2$ be $\varepsilon$-samples for $V_s, V_1, V_2$, respectively. Let $\textbf{U} = U_s \times U_1 \times U_2$. For every compatible pair $(f_1,f_2)$ of  monotone non-decreasing functions:
    \begin{gather*}
        \left| \textup{GFT}(M_{f_1,f_2},\textbf{V}) - \textup{GFT}(M_{f_1,f_2},\textbf{U}) \right| \leq 12\varepsilon
    \end{gather*}
\end{lemma}

\begin{proof}[Proof sketch] %
We next briefly explain the idea behind the proof of the lemma, its full proof can be found in Appendix \ref{app:missing7}. 
Denote by $S$ the union of the supports of $S_s,S_1,S_2$. Denote by $\lceil v \rceil_S$ and $\lfloor v \rfloor_S$ the rounding up and down of $v$ to the closest value in the $\varepsilon$-sample $S$, respectively. Additionally, denote by $\textup{GFT}_2(M_{f_1,f_2},\textbf{V})$ the expected GFT from trades that include buyer  $b_2$ ($s\rightarrow b_2$).

Now fix some $v_1$, and consider $\textup{GFT}_2(M_{f_1,f_2},\textbf{V})$ and $\textup{GFT}_2(M_{f_1,f_2},\textbf{U})$. For now, assume that the values of $f_2(v_1), f_2(\lceil v_1 \rceil_S),  f_2(\lfloor v_1 \rfloor_S)$ are `close' in the $\varepsilon$-sense; i.e. the probability of sampling a value between these values is at most $\varepsilon$. Therefore, with probability $1-\varepsilon$, the values of $v_s,v_2$ aren't `too close' to $f_2(v_1)$, and the same allocation of the mechanism will occur whether we round the value of $v_1$ before computing the price or not. If there is a significant difference between the values of $f_2(v_1), f_2(\lceil v_1 \rceil_S), f_2(\lfloor v_1 \rfloor_S)$ then the claim above does not necessarily hold. However, such significant jumps can only occur a small number of times, because $f_2$ is monotone and its range is in $[0,1]$. Therefore, from the properties of $\varepsilon$-samples, the probability of losing a trade in this case is also $O(\varepsilon)$.

Intuitively, this is exactly what happens in the $\varepsilon$-sample: the value of $v_1$ is rounded to some point in the sample, and the probability of that affecting the resulting allocation is negligible. Additionally, the expected GFT given that a specific allocation occurs is also not perturbed significantly, because we have an $\varepsilon$-sample for the participating agents so their expected value for the item is roughly the same. Thus, by taking the expectation over all values of $v_1$, and symmetrically doing the same for $v_2$ and considering the GFT from trades $s\rightarrow b_1$, results in an $O(\varepsilon)$ bound for these differences.
\end{proof}

Before using this lemma to reach our final result, we must explain how a mechanism that was learned on finite samples performs on the underlying distribution (that might be continuous). The mechanism we constructed is actually built from two functions $f_1^*,f_2^*:S\rightarrow S$, i.e. they are defined only on a finite set $S$. In order for them to be defined on the full domain $[0,1]$, we propose a simple rounding scheme: round all values $v_1\notin S$ down and all $v_2\notin S$ up, so we end up with $\lfloor v_1 \rfloor_S$ and $\lceil v_2 \rceil_S$. This means that for example, instead of setting the price for $s\rightarrow b_2$ to be $f_2(v_1)$ the mechanism sets it to $f_2(\lfloor v_1 \rfloor_S)$. Note also that the functions remain compatible, so the mechanism is still \nice.\footnote{Additionally, notice that from the way we chose to round the values, the functions are still tight. This actually isn't necessary, and different rounding schemes that preserve the compatibility are valid here as well.}

Using these results, we propose the following learning algorithm:

\begin{algorithm}[Learn an $\varepsilon$-Approximately Optimal Mechanism]\label{alg:learn}
    \textbf{Input:} Sample access to distributions $V_s, V_1, V_2$ supported on $[0,1]$, parameters $\varepsilon >0, \delta > 0$ \\
    \noindent \textbf{Output:} A mechanism that achieves, with probability $1-\delta$,  the optimal GFT on $V_s\times V_1\times V_2$ up to $O(\varepsilon)$\\
    \\
    \textbf{Step 1:} \texttt{Sample $O\left(\frac{1}{\varepsilon^2}\ln \frac{1}{\varepsilon\delta}\right)$ samples from each $V_s, V_1, V_2$ to create $S_s, S_1, S_2$}

    \textbf{Step 2:} \texttt{Run the GFT-Maximizing algorithm (Alg. \ref{alg:maxgft}) on $S_s, S_1, S_2$, get $f_1,f_2$}
        
    \textbf{Step 3:} \texttt{Return the mechanism $M_{f_1,f_2}$}
\end{algorithm}

In this algorithm we sample the distributions, and with probability at least $1-\delta$ the samples we generate are $\varepsilon$-samples. We then run Algorithm \ref{alg:maxgft} on the empirical distributions, and end up with functions that are associated with a GFT-optimal mechanism for the product of the empirical distributions. We compare any optimal mechanism %
for the true distributions \textbf{V} and the optimal mechanism we computed for \textbf{U}, the uniform  distributions over the samples. 
Since by Lemma \ref{lem:gftbound} if the samples are $\varepsilon$-samples both mechanisms achieve $O(\varepsilon)$-close GFT on both \textbf{V} and \textbf{U}, we get that our computed mechanism achieves close to optimal GFT on \textbf{V}.
Full proofs of these claims are given in the Appendix \ref{app:missing7}.

Summarizing the results from the previous few sections, we arrive at the conclusive theorem:

\thmlearn*

\printbibliography

\appendix

\section{Missing Proofs From Section \ref{sec:char}}

\subsection{Missing Proofs From Subsection 3.1}\label{app:missing31}
For completeness, we next present a proof that if $M$ is a \nice mechanism then the payment of a winning agent does not depend on his report.  

\begin{lemma}\label{lem:indi}
    Let $M=(A,p)$ be a \nice mechanism, then for every agent $i$ and $\forall v_i,v_i',\textbf{v}_{-i}$:
    \begin{gather*}
        \left( (A(v_i, \textib{v}_{-i}) \in W_i) \wedge (A(v_i', \textib{v}_{-i}) \in W_i) \right) \vee \left( (A(v_i, \textib{v}_{-i}) \notin W_i) \wedge (A(v_i', \textib{v}_{-i}) \notin W_i) \right)\\ 
        \Rightarrow p_i(v_i,\textbf{v}_{-i}) = p_i(v_i',\textbf{v}_{-i})
    \end{gather*}
\end{lemma}

\begin{proof}
    If $(A(v_i, \textib{v}_{-i}) \notin W_{i}) \wedge (A(v_i', \textib{v}_{-i}) \notin W_{i})$ then the claim holds trivially as $p_i(v_i,\textbf{v}_{-i}) = p_i(v_i',\textbf{v}_{-i}) = 0$. Otherwise, we know that in both cases the outcome is in the set of winning outcomes for agent $i$. Assume, towards a contradiction and w.l.o.g, that $\exists v_i,v_i',\textib{v}_{-i}: p_i(v_i,\textib{v}_{-i}) > p_i(v_i',\textib{v}_{-i})$. Then when agent $i$ has value $v_i'$:
    $$
        u_i(v_i, (v_i, \textib{v}_{-i})) = v_i - p_i(v_i,\textib{v}_{-i}) < v_i - p_i(v_i',\textib{v}_{-i}) = u_i(v_i, (v_i', \textib{v}_{-i}))
    $$
    and this violates the DSIC property.
\end{proof}

We use the above to show that if $M$ is a \nice mechanism then the payment of a winning buyer does not depend on his report or on the seller report.  
\begin{lemma}\label{lem:singleval}
    Let $M=(A,p)$ be a \nice mechanism. For every buyer $b$ and bidding profiles $(v_s,v_b,v_{-b}), (v_s',v_b',v_{-b})$, if $A(v_s,v_b,v_{-b}) = A(v_s',v_b',v_{-b}) = b$ then $p_b(v_s,v_b,v_{-b}) = p_b(v_s',v_b',v_{-b})$.
\end{lemma}

\begin{proof}
    Let $v_s,v_s',v_b,v_b',v_{-b}$. $M$ is SBB and normalized, so:
    $$
        \forall \textib{v}: A(\textib{v}) = b \Rightarrow p_b(\textib{v}) = -p_s(\textib{v})
    $$
    So:
    $$
        p_b(v_s,v_b,v_{-b}) \stackrel{(1)}{=} p_b(v_s,v_b',v_{-b}) \stackrel{(2)}{=} -p_s(v_s,v_b',v_{-b}) \stackrel{(3)}{=} -p_s(v_s',v_b',v_{-b}) \stackrel{(4)}{=} p_b(v_s',v_b',v_{-b})
    $$
    where 
    equality $(1)$ and equality $(3)$ follow by Lemma $\ref{lem:indi}$, and 
    equality $(2)$ and equality $(4)$ follow from the property we stated above.
\end{proof}

\begin{lemma}\label{lem:uniqdef}
    The functions in Definition \ref{def:payfuncs} are uniquely defined.
\end{lemma}

\begin{proof}
    Let $M=(A,p)$ be a \nice mechanism, and consider some bidding profile $\textib{v} = (v_s,v_1,v_2)$. For some buyer $b$, if $A(\textib{v}) = b$ then by Lemma \ref{lem:singleval} we know that for all reports $v_s,v_b$ s.t. $A(v_s,v_b,v_{-b}) = b$ the price is the same. Therefore there is a single value $p(\textib{v})$ and we set that to be the value of $f_b(v_{-b})$ by choosing some arbitrary $v_s,v_b$ s.t. $A(v_s,v_b,v_{-b}) = b$, so this is unique. Otherwise, there is no bidding profile $v_b,\textib{v}_{-b}$ s.t. $A(v_b,\textib{v}_{-i}) = b$ and so we set $f_b(v_{-b}) = \infty$, and that is uniquely defined as well.
\end{proof}

\begin{lemma}\label{lem:comppair}
    If a \nice mechanism $M=(A,p)$ is associated with a pair of functions $(f_1,f_2)$, then the pair $(f_1,f_2)$ is compatible.
\end{lemma}

\begin{proof}
    Towards a contradiction, assume that $\exists v_s,v_1,v_2$ s.t. $(v_1 > f_1(v_2)) \wedge (v_2 > f_2(v_1))$. By  Definition \ref{def:payfuncs} it holds that $f_1(v_2) = p_1(v_s,v_1,v_2)$. By Proposition  \ref{prop:classicmech}, if $v_1 > p_1(v_s,v_1,v_2) = f_1(v_2)$ then buyer $1$ wins. Similarly $v_2 > p_2(v_s,v_1,v_2) = f_2(v_1)$ and so buyer 2 wins as well. However, a mechanism is not feasible if both buyers trade concurrently. 
\end{proof}

\begin{proof}[Proof of Lemma \ref{lem:charfunc}]
    Let $f_1,f_2$ be functions defined as in Definition \ref{def:payfuncs}. From the way they are defined, if $A(\textib{v})=b_1$ then $p(\textib{v}) = (-f_1(v_2),f_1(v_2),0)$, and if $A(\textib{v})=b_2$, then $p(\textib{v}) = (-f_2(v_1),0,f_2(v_1))$. Additionally, by Lemma \ref{lem:comppair} this pair of functions is compatible. Finally, from Lemma \ref{lem:uniqdef} the functions are also uniquely defined.
\end{proof}

\subsection{Missing Proofs From Subsection 3.2}\label{app:missing32}

\begin{proof}[Proof of Lemma \ref{lem:compmech}]
    As we described previously, the notation of $p(\textib{v})$ describes the payment as such: if there is a trade $s\rightarrow b_i$, then $p_s(v_s,v_1,v_2)=-p_i(v_s,v_1,v_2), p_{-b_i}(v_s,v_1,v_2)=0$. Therefore, in all cases when there is a trade, the payment of the buyer equals the payment received by the seller, and the buyer that does not trade pays $0$. Additionally, when no trade occurs all of the payments are $0$. Therefore, the sum of payments is 0 at every bidding profile, and the mechanism is SBB. Furthermore, every losing agent pays $0$ so the mechanism is normalized.

    Additionally, as the allocation rule is monotone and payments are by critical values, by Proposition \ref{prop:classicmech} the mechanism is also DSIC. Moreover, since the mechanism is DSIC and normalized it is also ex-post IR. Finally, from the definition of $M_{f_1,f_2}$, if $A(\textib{v}) = 1$ then $p(\textib{v}) = (-f_1(v_2),f_1(v_2), 0)$ and if $A(\textib{v}) = 2$ then $p(\textib{v}) = (-f_2(v_1),0, f_2(v_1))$. Therefore $(f_1,f_2)$ are the associated pair of functions of  $M_{f_1,f_2}$.
\end{proof}

We will now prove that the mechanism defined in Definition \ref{def:m} maximizes the expected GFT amongst all mechanisms that share the same associated pair.

\begin{proof}[Proof of Lemma \ref{lem:mfmax}]
    Denote $M_{f_1,f_2} = (\hat{A},\hat{p})$. We will show that $\textup{GFT}(M_{f_1,f_2},\textib{v}) \geq \textup{GFT}(M,\textib{v})$ pointwise:
    \begin{enumerate}
        \item If $v_1 > f_1(v_2) \geq v_s$: From the definition of $f_1(v_2)$ we get $A(\textib{v}) = b_1 = \hat{A}(\textib{v})$, so $\textup{GFT}(M_{f_1,f_2},\textib{v}) = \textup{GFT}(M,\textib{v})$.
        \item Else, if $v_2 > f_2(v_1) \geq v_s$: From the definition of $f_2(v_1)$ we get $A(\textib{v}) = b_2 = \hat{A}(\textib{v})$, so $\textup{GFT}(M_{f_1,f_2},\textib{v}) = \textup{GFT}(M,\textib{v})$.
        \item Else, if $(v_1 = f_1(v_2) \geq v_s) \wedge (v_2 = f_2(v_1) \geq v_s$): $M_{f_1,f_2}$ obtains the maximal possible GFT in this case. %
        \item Else, if $v_1 = f_1(v_2) \geq v_s$: In this case either $M$ didn't trade at all and $\textup{GFT}(M,\textib{v})=0$ or it traded with $1$. Since $\hat{A}(\textib{v}) = b_1$ we get that in all cases $\textup{GFT}(M_{f_1,f_2},\textib{v}) \geq \textup{GFT}(M,\textib{v})$.
        \item Else, if $v_2 = f_2(v_1) \geq v_s$: In this case either $M$ didn't trade at all and $\textup{GFT}(M,\textib{v})=0$ or it traded with $2$. Since $\hat{A}(\textib{v}) = b_2$ we get that in all cases $\textup{GFT}(M_{f_1,f_2},\textib{v}) \geq \textup{GFT}(M,\textib{v})$.
        \item Else, $A(\textib{v}) = s = \hat{A}(\textib{v})$ and $\textup{GFT}(M_{f_1,f_2},\textib{v}) = 0 = \textup{GFT}(M,\textib{v})$.
    \end{enumerate}
\end{proof}

\section{Missing Proofs From Section \ref{sec:optimal}}\label{app:missing5}

\subsection{Upper Semi-Continuity of GFT}\label{app:usc}

\begin{definition} [Bilateral GFT at price $p$]
    Let $V_s,V_b$ be two distributions over the valuations of the seller and the buyer. We define the \textup{Bilateral Gains From Trade at price $p$} as:
    \begin{gather*}
        \textup{GFT}(p, V_s\times V_b) = \mathbb{E}[V_b - V_s | V_b \geq p \geq V_s] \cdot Pr[V_b \geq p \geq V_s]
    \end{gather*}
\end{definition}

{For point $x_0$ and $\delta>0$, let  $B(x_0;\delta)$ be a ball of radius $\delta$ around $x_0$.}
\begin{definition}
    Let $f: D \rightarrow \mathbb{R}$ and let $x_0\in D$. We say that $f$ is \emph{upper semi-continuous at $x_0$} if $\forall \varepsilon >0\  \exists \delta >0$ s.t.:
    \begin{gather*}
        \forall x\in B(x_0;\delta)\cap D \quad f(x) < f(x_0) + \varepsilon
    \end{gather*}
    Additionally, we say that \emph{$f$ is upper semi-continuous} if it is upper semi-continuous at every $x\in D$.
\end{definition}

\begin{theorem}[Continuity from above, Theorem 2.1 in \cite{billingsley1995probability}]
    Let $(\Omega, \mathcal{F}, P)$ be a probability space, and let $A_1 \supseteq A_2 \supseteq \ldots $ be 
    a sequence of subsets of $\Omega$ ($\forall i\ A_i\subseteq \Omega$).
    Define:
    \begin{gather*}
        A = \bigcap_{n=1}^\infty A_n
    \end{gather*}
    then $lim_{n\rightarrow \infty}P(A_n) = P(A)$.
\end{theorem}

\begin{corollary}\label{cor:local}
    Let $(\Omega, \mathcal{F}, P)$ be a probability space, and let $A_1 \supseteq A_2 \supseteq \ldots $ be 
    a sequence of subsets of $\Omega$ ($\forall i\ A_i\subseteq \Omega$).
    Then $\lim_{n\rightarrow\infty}P(A_n) - P(A) = 0$. Additionally, if $\Omega \subseteq \mathbb{R}^k$ then:
    \begin{gather*}
        \forall \omega \in \Omega,\varepsilon > 0 \ \exists\delta > 0\text{ s.t. }P(B(\omega;\delta)\cap \Omega) - P(\omega) < \varepsilon
    \end{gather*}
\end{corollary}

\begin{lemma}\label{gftusc}
    For any two  bounded distributions $V_s,V_b$, the function $\textup{GFT}(\cdot, V_s\times V_b)$ is upper semi-continuous.
\end{lemma}

\begin{proof}
    W.l.o.g assume $V_s,V_b$ are supported on  subsets of $[0,1]$. Denote by $P_s, P_b$ the probability measures on the values of  $V_s$ and $V_b$, respectively. Let $p_0\in [0,1]$, and let $\varepsilon > 0$.  For $\varepsilon > 0$, let $\delta_s>0$ be such that   $P_s(B(p_0;\delta_s)\cap [0,1]) - P_s(p_0) < \varepsilon$, and let $\delta_b>0 $ be such that  $P_b(B(p_0;\delta_b)\cap [0,1]) - P_b(p_0) < \varepsilon$ (such values exist by Corollary \ref{cor:local}). Let $\delta = \min \{\delta_s,\delta_b\}$. Let $p \in B(p_0;\delta) \cap [0,1]$. We prove the claim by considering the following two cases:
    \begin{itemize}
        \item $p < p_0$:
        \begin{gather*}
            \textup{GFT}(p, V_s\times V_b)- \textup{GFT}(p_0, V_s\times V_b) = \\
            \mathbb{E}[V_b-V_s | V_b \geq p \geq V_s]\cdot Pr[V_b \geq p \geq V_s] - \mathbb{E}[V_b-V_s | V_b \geq p_0 \geq V_s]\cdot Pr[V_b \geq p_0 \geq V_s] \\
             = \mathbb{E}[V_b-V_s | p_0 >V_b \geq p \geq V_s]\cdot Pr[p_0 > V_b \geq p \geq V_s] \\
             - \mathbb{E}[V_b-V_s | V_b \geq p_0 \geq V_s > p]\cdot Pr[V_b \geq p_0 \geq V_s > p] \\
             \leq Pr[p_0 > V_b \geq p \geq V_s] = Pr[p_0 > V_b \geq p | p \geq V_s] \cdot Pr[p \geq V_s] \\ \leq  Pr[p_0 > V_b \geq p] \leq P_b(B(p_0;\delta)\cap [0,1]) - P_b(p_0) < \varepsilon
        \end{gather*}

        \item $p > p_0$:
        \begin{gather*}
            \textup{GFT}(p, V_s\times V_b) - \textup{GFT}(p_0, V_s\times V_b) = \\
            \mathbb{E}[V_b-V_s | V_b \geq p \geq V_s]\cdot Pr[V_b \geq p \geq V_s]- \mathbb{E}[V_b-V_s | V_b \geq p_0 \geq V_s]\cdot Pr[V_b \geq p_0 \geq V_s] \\
             = \mathbb{E}[V_b-V_s | V_b \geq p \geq V_s > p_0]\cdot Pr[V_b \geq p \geq V_s > p_0] \\
             - \mathbb{E}[V_b-V_s | p >V_b \geq p_0 \geq V_s]\cdot Pr[p > V_b \geq p_0 \geq V_s] \\
             \leq Pr[V_b \geq p \geq V_s > p_0] = Pr[p \geq V_s > p_0 | V_b \geq p] \cdot Pr[V_b \geq p] \\
             \leq Pr[p \geq V_s > p_0] \leq P_s(B(p_0;\delta)\cap [0,1]) - P_s(p_0) < \varepsilon
        \end{gather*}
    \end{itemize}

    So $\textup{GFT}(\cdot, V_s\times V_b)$ is upper semi-continuous at every point $p_0$, and therefore it is upper semi-continuous.
\end{proof}
\subsection{Restricted Best Price}
We next show that in Definition \ref{def:resbp}
of the Restricted Best Price, the supremum is  actually a maximum:

\begin{lemma}\label{lem:argmax}
     Let $V_s,V_b$ be two bounded distributions, then $\exists p^*\in \mathbb{R}$ s.t. $\forall p\in \mathbb{R}: \textup{GFT}(p^*, V_s\times V_b) \geq \textup{GFT}(p, V_s\times V_b)$.
\end{lemma}

\begin{proof}
    By Lemma \ref{gftusc} the function $\textup{GFT}(\cdot, V_s\times V_b)$ is upper semi-continuous. The function is also bounded, as the distributions are bounded, so from the extreme value theorem it attains its supremum.
\end{proof}

\begin{lemma}
    Let $V_s,V_b$ be bounded distributions, and let $p^*$ be the best price - then $p^*$ can be attained.
\end{lemma}

\begin{proof}
    By Lemma \ref{lem:argmax} we get that $\argmax\limits_{p\in \mathbb{R}} \textup{GFT}(p, V_s\times V_b)$ is well defined. Now, towards a contradiction, assume that there is no $\max\left\{\argmax\limits_{p\in \mathbb{R}} \textup{GFT}(p, V_s\times V_b)\right\}$. Therefore, we have an infinite increasing series of $p_1 < p_2 < \ldots$ s.t. $\textup{GFT}(p_1, V_s\times V_b) = \textup{GFT}(p_2, V_s\times V_b) = \ldots$. Again using the fact that GFT is upper semi-continuous at $p$, we know that the supremum of this sequence is attainable, and that there is a $\max\left\{\argmax\limits_{p\in \mathbb{R}} \textup{GFT}(p, V_s\times V_b)\right\}$.
\end{proof}

\begin{lemma}\label{lem:rbpmono}
    Let $V_s,V_b$ be bounded distributions, then the Restricted Best Price is monotone non-decreasing in $r$, i.e. $\forall r_1 \leq r_2: p^*(r_1, V_s\times V_b) \leq p^*(r_2, V_s\times V_b)$.
\end{lemma}

\begin{proof}
    If $p^*(r_1, V_s\times V_b) < r_2$, then trivially $p^*(r_1, V_s\times V_b) < r_2 \leq p^*(r_2, V_s\times V_b)$. Otherwise, $p^*(r_1, V_s\times V_b) \geq r_2$, and also $p^*(r_1, V_s\times V_b) = \max\left\{\argmax\limits_{p \geq r_1} \textup{GFT}(p, V_s\times V_b)\right\}$. Therefore, there is no $p > p^*(r_1, V_s\times V_b)$ s.t. $\textup{GFT}(p, V_s\times V_b) \geq \textup{GFT}(p^*(r_1, V_s\times V_b), V_s\times V_b)$. Therefore, $p^*(r_1, V_s\times V_b) = p^*(r_2, V_s\times V_b)$.
\end{proof}

\subsection{Proofs for the Modification Steps}\label{app:mod}

\begin{lemma}\label{lem:resiff}
     $\forall v_1, v_2: f_1(v_2) \geq r_1^{f_2}(v_2) \iff  f_2(v_1) \geq r_1^{f_2}(v_1)$.
\end{lemma}

\begin{proof}
    \begin{gather*}
        \forall v_1, v_2: f_1(v_2) \geq r_1^{f_2}(v_2) = \sup \{v_1 | v_2 \geq f_2(v_1)\} \\
        \iff \forall v_1, v_2: v_1 \geq f_1(v_2) \Rightarrow v_2 \leq f_2(v_1) \\
        \iff \forall v_1, v_2: v_2 \geq f_2(v_1) \Rightarrow v_1 \leq f_1(v_2) \\
        \iff  \forall v_1, v_2: f_2(v_1) \geq r_2^{f_1}(v_1) = \sup \{v_2 | v_1 \geq f_1(v_2)\}
    \end{gather*}
\end{proof}

\begin{lemma}\label{lem:compres}
    If $\forall v_1, v_2: f_1(v_2) \geq r_1^{f_2}(v_2)$ (or $f_2(v_1) \geq r_2^{f_1}(v_1)$, by Lemma \ref{lem:resiff}) then the pair $(f_1,f_2)$ is compatible.
\end{lemma}

\begin{proof}
    Contrariwise and w.l.o.g assume there are $v_1,v_2$ s.t. $(v_1 > f_1(v_2)) \wedge (v_2 > f_2(v_1))$. However, $f_1(v_2) \geq r_1^{f_2}(v_2) = \sup \{v_1 | v_2 \geq f_2(v_1)\} \geq v_1$ - and this is a contradiction.
\end{proof}

\begin{lemma}\label{lem:starcomp}
     Let $(f_1,f_2)$ be a pair of compatible functions, then each one of the pairs $(g(f_2,V_s\times V_1),{f}_2)$, $({f}_1,g(f_1,V_s\times V_2))$ is compatible.
\end{lemma}

\begin{proof}
    Denote $\tilde{f}_2 = g(f_1,V_s\times V_2)$. From the way we defined $\tilde{f}_2$ we know that $\forall v_1, v_2: \tilde{f}_2(v_1) \geq r_2^{f_1}(v_1)$, and so by Lemma \ref{lem:compres} the pair $(f_1, \tilde{f}_2)$ is compatible. The reasoning for the other pair is symmetric.
\end{proof}

\begin{lemma}\label{lem:starmono}
     Let $(f_1,f_2)$ be a pair of compatible functions, then the functions $g(f_1,V_s\times V_2)$ and $g(f_2,V_s\times V_1)$ are each monotone non-decreasing.
\end{lemma}

\begin{proof}
    By Lemma \ref{lem:rbpmono} the restricted best price is monotone non-decreasing, and as we stated the compatibility restriction is also monotone non-decreasing. The functions $g(f_1,V_s\times V_2)$ and $g(f_2,V_s\times V_1)$ are a composition of these functions, and therefore are also monotone non-decreasing.
\end{proof}

\begin{proof}[Proof of Lemma \ref{lem:starcompmono}]
    Immediate by Lemmata \ref{lem:starcomp}, \ref{lem:starmono}.
\end{proof}

\begin{proof}[Proof of Lemma \ref{lem:moregft}]
    We prove the second inequality - that $\textup{GFT}(M_{{f}_1,\tilde{f}_{2}},\textbf{V}) \geq \textup{GFT}(M_{f_1,f_2},\textbf{V})$, the proof of the first is similar. Split the GFT to $\textup{GFT}_1$ and $\textup{GFT}_2$, as in Notation \ref{not:gft12}. For $\textup{GFT}_2$ and $\forall v_1$ we know that $\textup{GFT}(\tilde{f}_2(v_1), V_s\times V_2) \geq \textup{GFT}(f_2(v_1), V_s\times V_2)$, because $\tilde{f}_2(v_1)$ can be the same $f_2(v_1)$ or improve on it. 
    
    Since we didn't change $f_1$, the $\textup{GFT}_1$ is the same in both sides of the inequality except in cases of ties. In the case of a tie all that remains here is to show that we didn't lose any GFT from changing $f_2$ to $\tilde{f}_2$. The only way to lose $\textup{GFT}_1$ in this case is if we had a tie that $b_1$ won previously but now loses. This can occur only if at $v_1,v_2$ we have $f_2(v_2) \leq f_1(v_2) < \tilde{f}_2(v_1)$. However in this case what we lose in $\textup{GFT}_1$ we gain in $\textup{GFT}_2$ (and maybe even more), so the claim holds.
    
    Therefore, summing all cases we get that $\textup{GFT}(M_{{f}_1,\tilde{f}_{2}},\textbf{V}) \geq \textup{GFT}(M_{f_1,f_2},\textbf{V})$.
\end{proof}

\section{Missing Proofs From Section \ref{sec:learn}}\label{app:missing7}

We begin by introducing several well-established definitions and outcomes concerning learning from samples.

\begin{definition}[Def 14.1 from \cite{mitzenmacher2017probability}]
    A \textup{range space} is a pair $(X, \mathcal{R})$ where:
    \begin{enumerate}
        \item $X$ is a (finite or infinite) set of \textup{points};
        \item $\mathcal{R}$ is a family of subsets of $X$, called \textup{ranges}.
    \end{enumerate}
\end{definition}

In this paper we use the range space where $X=[0,1]$ and $\mathcal{R}$ is the family of all closed intervals $[a,b] \subseteq [0,1]$. It is known that the VC dimension of this range space is $2$ (\cite{mitzenmacher2017probability}, page 364).

\begin{lemma}\label{lem:sampleexp}
    Let $V$ be a distribution over %
    $[0,1]$, and let $S$ be an $\varepsilon$-sample of $V$. Then $\forall R = [a,b] \subseteq [0,1]: | \mathbb{E}_{v\sim V}[v\cdot \mathbb{I}_v\{R\}] - \mathbb{E}_{s\sim \textup{U}(S)}[s\cdot \mathbb{I}_s\{R\}] | \leq \varepsilon$, where $\mathbb{I}_v\{R\}$ is $1$ if $v\in R$ and $0$ otherwise. 
\end{lemma}

\begin{proof}
    Note that by the definition of expected value:
    \begin{align*}
        &\mathbb{E}_{v\sim V}[v\cdot \mathbb{I}_v\{R\}] = \int_R x\, dV &\mathbb{E}_{s\sim U(S)}[s\cdot \mathbb{I}_s\{R\}] = \int_R x\, dU(S)
    \end{align*}

    Taking the absolute value of the difference:
    \begin{gather*}
        \left| \mathbb{E}_{v\sim V}[v\cdot \mathbb{I}_v\{R\}] - \mathbb{E}_{s\sim \textup{U}(S)}[s\cdot \mathbb{I}_s\{R\}] \right| = \left| \int_R x\, dV - \int_R x\, dU(S) \right| \\
        = \left| \int_R x\, (dV-dU(S)) \right|
        \stackrel{(1)}{\leq} \left| \int_R x\cdot \varepsilon \right| \stackrel{(2)}{\leq} \varepsilon
    \end{gather*}
    where $(1)$ is by $\varepsilon$-sample properties, and $(2)$ is because $R \subseteq [0,1]$.
\end{proof}

\begin{theorem}[Thm. 14.15 from \cite{mitzenmacher2017probability}]\label{thm:esample}
    Let $(X, \mathcal{R})$ be a range space with VC dimension $d$ and let $\mathcal{D}$ be a probability distribution on $X$. For any $\varepsilon > 0, 1/2 > \delta>0$, there is an
    \begin{gather*}
        m = O\left(\frac{d}{\varepsilon^2}\ln \frac{d}{\varepsilon} + \frac{1}{\varepsilon^2}\ln\frac{1}{\delta}\right)
    \end{gather*}
    such that a random sample from $\mathcal{D}$ of size greater than or equal to $m$ is an $\varepsilon$-sample for $X$ with probability at least $1-\delta$.
\end{theorem}

Specifically, in the range space we use where $X=[0,1]$ and $\mathcal{R}$ is the family of all closed intervals $[a,b] \subseteq [0,1]$ the VC dimension is 2, and therefore there is a sample of size $O\left(\frac{1}{\varepsilon^2}\ln \frac{1}{\varepsilon\delta}\right)$ which is an $\varepsilon$-sample with probability at least $1-\delta$. Next, we notate the separate parts of GFT that we gain from each buyer:

\begin{notation}\label{not:gft12}
    Let $M=(A,p)$ be a DSIC mechanism for the 1-seller 2-buyer setting. We denote $M$'s expected GFT for prior  $\textbf{V}$ generated by trades between the seller and buyer $b$ by
$\textup{GFT}_b(M,\textbf{V})$:
$$\textup{GFT}_b(M,\textbf{V}) = \mathbb{E}_{\textib{v}\sim\textbf{V}}\left[(v_b-v_s)\cdot \mathbb{I}\{A(\textib{v}) = b\}\right] 
$$
    Note that $\textup{GFT}(M,\textbf{V}) = \textup{GFT}_1(M,\textbf{V}) + \textup{GFT}_2(M,\textbf{V})$.
\end{notation}

The following lemma contains the essence of our ability to achieve robust approximations using $\varepsilon$-samples. On an intuitive level, it demonstrates that rounding to the worst-case scenario can only result in a loss of $\varepsilon$ GFT during each iteration. This is due to the fact that the potential points within $S$ encompass a significant portion of the distribution's weight - so the probability to `fall through the cracks' and lose GFT is negligible. The next lemma immediately implies Lemma \ref{lem:gftbound}.

\begin{lemma}\label{lem:gftbound12}
    Let $\textbf{V}=V_s\times V_1 \times V_2$ be a product distribution over %
    $[0,1]^3$, let $S_s, S_1, S_2$ be $\varepsilon$-samples for $V_s, V_1, V_2$, respectively. Let $\textbf{U} = U_s \times U_1 \times U_2$. For every compatible pair $(f_1,f_2)$ of monotone non-decreasing functions, and every buyer $b\in \{1,2\}$:
    \begin{gather*}
        \left| \textup{GFT}_b(M_{f_1,f_2},\textbf{V}) - \textup{GFT}_b(M_{f_1,f_2},\textbf{U}) \right| \leq 6\varepsilon 
    \end{gather*}
\end{lemma}

\begin{proof}
    We will prove this for $\textup{GFT}_2$ and the proof for $GFT_1$ is symmetric. First we must define how to round up or down to a given sample $S$:

    \begin{definition}[rounding to $S$]
        Let $S$ be a finite set, $\{0,1\} \subseteq S$, for every $x\in [0,1]$ the  rounding of $x$ up and down w.r.t. $S$ are defined as follows:
        \begin{gather*}
            \lceil x \rceil_{S} = \min_{s\in S} \{s \geq x\} \qquad \qquad
            \lfloor x \rfloor_{S} = \max_{s\in S} \{s \leq x\} 
        \end{gather*}
    \end{definition}
    
    Now let $S$ be the union of the supports of $S_s,S_1,S_2$ with $\{0,1\}$, and denote its elements by $S = \{0=s_1 < \ldots < s_k = 1 \}$. Fix $(s_{i-1},s_i]$ and let $v_1 \in (s_{i-1},s_i]$:
    \begin{gather}
        \textup{GFT}_2(M_{f_1,f_2},\textbf{V} | v_1) \stackrel{(a)}{=} \mathbb{E}_{v_2\sim V_2, v_s\sim V_s}[v_2 - v_s | v_2 \geq f_2(v_1) \geq v_s]\cdot \textup{Pr}_{v_2\sim V_2, v_s\sim V_s}[v_2 \geq f_2(v_1) \geq v_s] \nonumber \\
        \stackrel{(b)}{=} \mathbb{E}_{v_2\sim V_2}[v_2 | v_2 \geq f_2(v_1)] \cdot \textup{Pr}_{v_2\sim V_2}[v_2 \geq f_2(v_1)] - \mathbb{E}_{v_s\sim V_s}[v_s | f_2(v_1) \geq v_s] \cdot \textup{Pr}_{v_s\sim V_s}[f_2(v_1)\geq v_s] \nonumber \\ 
        \stackrel{(c)}{=} \mathbb{E}_{v_2\sim V_2}[v_2 | v_2 \geq f_2(\lceil v_1 \rceil_S)] \cdot \textup{Pr}_{v_2\sim V_2}[v_2 \geq f_2(\lceil v_1 \rceil_S)] \\
        + \mathbb{E}_{v_2\sim V_2}[v_2 | f_2(\lceil v_1 \rceil_S) > v_2 \geq f_2(v_1)]  \cdot \textup{Pr}_{v_2\sim V_2}[f_2(\lceil v_1 \rceil_S) > v_2 \geq f_2(v_1)] \\ 
        - \mathbb{E}_{v_s\sim V_s}[v_s | f_1(v_2) \geq v_s > f_2(\lfloor v_1 \rfloor_S)] \cdot \textup{Pr}_{v_s\sim V_s}[f_1(v_2) \geq v_s > f_2(\lfloor v_1 \rfloor_S)] \\
        - \mathbb{E}_{v_s\sim V_s}[v_s | f_2(\lfloor v_1 \rfloor_S) \geq v_s]  \cdot \textup{Pr}_{v_s\sim V_s}[f_2(\lfloor v_1 \rfloor_S) \geq v_s]
    \end{gather} %
    where $(a)$ follows from the definition of GFT, $(b)$ from independence of the random variables, and $(c)$ is simply splitting the equation to its different parts.

    Consider the elements $(2)$ and $(3)$ in the equation above. In both of them, we have the expectation that is $\leq 1$, times the probability that a variable is between the value of $f_2(v_1)$ and a rounding of it (either $f_2(\lceil v_1 \rceil_S$ or $f_2(\lfloor v_1 \rfloor_S$). However, note that one side of that inequality is strict in both of them, and it is specifically to the side of the rounding. Therefore, if $v_1 \in S$, the probability of this is exactly 0. Otherwise, we have the probability that the variable $v_1$ is between two elements in $S$, but is not itself in $S$. Since $S_1$ is an $\varepsilon$-sample, the probability of that is at most $\varepsilon$.

    We calculate the upper bound on $\textup{GFT}_2$ of $M_{f_1,f_2}$ on $\textbf{V}$, but only on trades when %
    $v_1 \in [s_{i},s_{i+1})$:
    \begin{gather*}
        \textup{GFT}_2(M_{f_1,f_2},\textbf{V} | v_1\in (s_{i-1}, s_i]) \\
        \stackrel{(a)}{\leq} \textup{Pr}_{v_1\in V_1}[v_1\in (s_{i-1}, s_i]] \cdot \Big(\mathbb{E}_{v_2\sim V_2}[v_2 | v_2 \geq f_2(\lfloor v_1 \rfloor_S)] \cdot \textup{Pr}_{v_2\sim V_2}[v_2 \geq f_2(\lfloor v_1 \rfloor_S)] \\
        - \mathbb{E}_{v_s\sim V_s}[v_s | f_2(\lfloor v_1 \rfloor_S) \geq v_s]  \cdot \textup{Pr}_{v_s\sim V_s}[f_2(\lfloor v_1 \rfloor_S) \geq v_s]\Big) \\
        \stackrel{(b)}{\leq} \left(\textup{Pr}_{u_1\in U_1}[u_1 = s_i]+\varepsilon\right) \cdot \Bigl((\mathbb{E}_{u_2\sim U_2}[u_2 | u_2 \geq f_2(s_i)]+\varepsilon) \cdot (\textup{Pr}_{u_2\sim U_2}[u_2 \geq f_2(s_i)]+\varepsilon) \\
        - (\mathbb{E}_{u_s\sim U_s}[u_s | f_2(s_i) \geq v_s]-\varepsilon)  \cdot (\textup{Pr}_{u_s\sim U_s}[f_2(s_i) \geq v_s]-\varepsilon) \Bigr) 
       \\
        \stackrel{(c)}{\leq} \textup{Pr}_{u_1\in U_1}[u_1 = s_i] \cdot \Big(\mathbb{E}_{u_2\sim U_2}[u_2 | u_2 \geq f_2(s_i)] \cdot (\textup{Pr}_{u_2\sim U_2}[u_2 \geq f_2(s_i)] \\
        - \mathbb{E}_{u_s\sim U_s}[u_s | f_2(s_i) \geq v_s]  \cdot \textup{Pr}_{u_s\sim U_s}[f_2(s_i) \geq v_s] \Big) + 6\varepsilon \\
        \stackrel{(d)}{=} \textup{GFT}_2(M_{f_1,f_2},\textbf{U} | u_1 = s_i) + 6\varepsilon
    \end{gather*}
    where the inequality $(a)$ only increases the GFT by lowering the price for the buyer to $f_2(\lfloor v_1 \rfloor_S)$ (this includes both $(1)$ and $(2)$ above) and disregarding the $-\varepsilon$ for $(3)$; $(b)$ follows from $\varepsilon$-sample properties, Lemma \ref{lem:sampleexp}, and taking the most extreme options for the signs of $\varepsilon$; $(c)$ is just simplification of the elements; and $(d)$ follows from the definition of GFT.

    Similarly, we calculate the lower bound on $\textup{GFT}_2$ of $M_{f_1,f_2}$ on $\textbf{V}$, but only on trades when  $v_1 \in [s_{i},s_{i+1})$:
    \begin{gather*}
        \textup{GFT}_2(M_{f_1,f_2},\textbf{V} | v_1\in (s_{i-1}, s_i])  \\
        \stackrel{(a)}{\geq} \textup{Pr}_{v_1\in V_1}[v_1\in (s_{i-1}, s_i]] \cdot \big(\mathbb{E}_{v_2\sim V_2}[v_2 | v_2 \geq f_2(\lceil v_1 \rceil_S)] \cdot \textup{Pr}_{v_2\sim V_2}[v_2 \geq f_2(\lceil v_1 \rceil_S)] \\ \displaybreak
        - \mathbb{E}_{v_s\sim V_s}[v_s | f_2(\lceil v_1 \rceil_S) \geq v_s]  \cdot \textup{Pr}_{v_s\sim V_s}[f_2(\lceil v_1 \rceil_S) \geq v_s] \big) \\
        \stackrel{(b)}{\geq} (\textup{Pr}_{u_1\in U_1}[u_1 = s_i]-\varepsilon) \cdot \big((\mathbb{E}_{u_2\sim U_2}[u_2 | u_2 \geq f_2(s_i)]-\varepsilon) \cdot (\textup{Pr}_{u_2\sim U_2}[u_2 \geq f_2(s_i)]-\varepsilon) \\
        - (\mathbb{E}_{u_s\sim U_s}[u_s | f_2(s_i) \geq v_s]+\varepsilon)  \cdot (\textup{Pr}_{u_s\sim U_s}[f_2(s_i) \geq v_s]+\varepsilon) \big) \\
        \stackrel{(c)}{\geq} \textup{Pr}_{u_1\in U_1}[u_1 = s_i] \cdot \big(\mathbb{E}_{u_2\sim U_2}[u_2 | u_2 \geq f_2(s_i)] \cdot (\textup{Pr}_{u_2\sim U_2}[u_2 \geq f_2(s_i)] \\
        - \mathbb{E}_{u_s\sim U_s}[u_s | f_2(s_i) \geq v_s]  \cdot \textup{Pr}_{u_s\sim U_s}[f_2(s_i) \geq v_s] \big) - 6\varepsilon \\
        \stackrel{(d)}{=} \textup{GFT}_2(M_{f_1,f_2},\textbf{U} | u_1 = s_i) - 6\varepsilon
    \end{gather*}
    where for $(a)$ we just removed $(2)$ because we are bounding from below and increased to price to $f_2(\lceil v_1 \rceil_S)$ so we include $(3)$ and $(4)$; and $(b),(c),(d)$ are for the same reasons as above.

    Summing over all intervals $(s_{i-1},s_i]$ and taking the integral over $V_1$, we get the desired bound for $\textup{GFT}_2$. Similar calculations hold for getting the desired bound for $\textup{GFT}_1$.
\end{proof}

\begin{proof}[Proof of Lemma \ref{lem:gftbound}]
    Immediate by Lemma \ref{lem:gftbound12}.
\end{proof}

Using the result of the lemma and the algorithm from the previous section, we claim the following:

\begin{lemma}\label{lem:algbound}
    Let $\textbf{V}=V_s\times V_1 \times V_2$ be a product distribution over %
    $[0,1]^3$, let $S_s, S_1, S_2$ be $\varepsilon$-samples for $V_s, V_1, V_2$, respectively. Let $\textbf{U} = U_s \times U_1 \times U_2$. Let $M_U^*$ be a GFT-optimal mechanism on $\textbf{U}$. For every \nice mechanism $M$:
     \begin{gather*}
         \textup{GFT}(M,\textbf{V}) \leq \textup{GFT}(M_U^*,\textbf{V}) + 24\varepsilon
     \end{gather*}
\end{lemma}

\begin{proof}
    Since $M$ and $M_U^*$ are \nice, by Lemma \ref{lem:charfunc} we know that each has a  pair of functions associated with it. Denote the associated pair of $M$ by $(f_1,f_2)$ and the associated pair of $M_U^*$ by $(f_1^U,f_2^U)$. Note that since $M_U^*$ is optimal then $\textup{GFT}(M_U^*,\textbf{V}) = \textup{GFT}(M_{f_1^U,f_2^U},\textbf{V})$. Then:
    \begin{gather*}
        \textup{GFT}(M,\textbf{V}) - \textup{GFT}(M_U^*,\textbf{V}) 
        \stackrel{(a)}{\leq} \textup{GFT}(M_{f_1,f_2},\textbf{V}) - \textup{GFT}(M_{f_1^U,f_2^U},\textbf{V}) \\
        \stackrel{(b)}{=} \textup{GFT}(M_{f_1,f_2},\textbf{V}) - \textup{GFT}(M_{f_1,f_2},\textbf{U}) \\
        + \textup{GFT}(M_{f_1,f_2},\textbf{U}) - \textup{GFT}(M_{f_1^U,f_2^U},\textbf{U}) \\
        +\textup{GFT}(M_{f_1^U,f_2^U},\textbf{U}) - \textup{GFT}(M_{f_1^U,f_2^U},\textbf{V}) \\
        \stackrel{(c)}{\leq} \textup{GFT}(M_{f_1,f_2},\textbf{V}) - \textup{GFT}(M_{f_1,f_2},\textbf{U}) 
        + \textup{GFT}(M_{f_1^U,f_2^U},\textbf{U}) - \textup{GFT}(M_{f_1^U,f_2^U},\textbf{V}) \stackrel{(d)}{\leq} 24\varepsilon
     \end{gather*}
     where $(a)$ is by Lemma \ref{lem:mfmax} and from the fact that 
     $\textup{GFT}(M_U^*,\textbf{V}) = \textup{GFT}(M_{f_1^U,f_2^U},\textbf{V})$,
     $(b)$ we added and subtracted the same elements, $(c)$ is due to Theorem \ref{thm:algmax} that $\textup{GFT}(M_{f_1,f_2}, \textbf{U})$ $ \leq \textup{GFT}(M_{f_1^U,f_2^U},\textbf{U})$ because $M_{f_1^U,f_2^U}$ maximizes GFT on $\textbf{U}$, and $(d)$ is from using the bound by Lemma \ref{lem:gftbound} on $f_1,f_2$ and $f_1^U,f_2^U$.
\end{proof}

\begin{proof}[Proof of Theorem \ref{thm:learn}] %
    The algorithm first picks samples     
    $S_s, S_1, S_2$, each of size $O\left(\frac{1}{\varepsilon^2}\ln \frac{1}{\varepsilon\delta}\right)$ from 
    $V_s, V_1, V_2$, respectively.
    Let $M$ be a GFT-optimal mechanism for the empirical distribution $\textbf{U} = U_s \times U_1 \times U_2$ (computed in poly-time using Theorem \ref{thm:algmax}).
    By Theorem \ref{thm:esample}, with probability at least $1-\delta$, the multisets $S_s, S_1, S_2$ are $\varepsilon$-samples of the corresponding distributions.
    Thus, if $S_s, S_1, S_2$ are indeed $\varepsilon$-samples then by Lemma \ref{lem:algbound} the expected GFT of $M$ on $\textib{V}$ is at most $24\varepsilon$ less than the expected GFT of any other mechanism on $\textib{V}$ -- specifically of $\OPTniceV$.
\end{proof}

\end{document}

%% file: intro.tex
\section{Introduction}

In principle, designing a market for identical goods should be an easy task: the market should 
aggregate  the demand from all buyers as well as the supply from all the sellers, 
compute the market
clearing price, and have the sellers whose value lies below the clearing price 
sell to the buyers whose
value lie above the clearing price.  
This maximizes gains from trade 
and achieves optimal social welfare.\footnote{The ``gains from trade'' are defined as the increase in total welfare due to trade.  Optimizing GFT is equivalent to optimizing
social welfare.  Approximating GFT in a multiplicative sense may be harder, but in
this paper we focus on additive approximation for which they are equivalent.} Yet, this solution does not take incentives into account:
there is an incentive for sellers to overbid and for buyers to underbid. While in ``large'' markets 
one may expect these effects to be negligible, ``small'' markets 
must be carefully designed, taking the agents' strategic considerations into account.

The celebrated VCG mechanism can be used to handle the incentives issue by inducing truthful reporting as dominant strategies.
Unfortunately, in trade settings the VCG mechanism is not budget balanced but rather loses money, 
and thus is not feasible without an external source of subsidies. 
As the VCG mechanism is essentially the unique mechanism that maximizes gains from trade in dominant strategies, this implies the general impossibility of maximizing gains from trade in a budget balanced way in dominant strategies.
In fact, \cite{myerson1983efficient} 
show that maximizing gains from trade without subsidies is impossible even in the
much more relaxed sense of Bayesian incentive compatibility.\footnote{And
even if we only require interim
individual rationality and weak budget balance.} Furthermore, 
this holds even for the setting of a bilateral trade in which a single seller has one item to sell to a single buyer.

Given that it is impossible to fully maximize gains from trade in strategic trade settings without subsidies, we aim to computationally design strategic budget-balanced  mechanisms that have the {\em highest possible} gains from trade.  
We focus on the simplest class of mechanisms:

\begin{definition}
   A market mechanism  
   is called {\em \nice} if it is: 
(1) Deterministic 
(2) Dominant Strategy Incentive Compatible (DSIC)
(3) Normalized (i.e. a participant that does not trade neither pays nor gets any money.)
(4) Ex-post Strongly budget balanced (SBB).  I.e., in every instance the payment by a trading buyer is equal to the amount received by the seller
(5) Ex-post individually rational  (i.e., in no instance does any truthful participant get negative utility.) 
\end{definition}

Our main focus is on a setting where we are given samples from an unknown distribution of values and our goal is to design a mechanism that has good performance on
the underlying distribution.  
This setting is naturally viewed as a learning procedure that finds a \nice mechanism with close-to-optimal gains from trade, when the learner only has access to ``historical data'' of samples from the underlying distribution of values.  Our work thus fits into the recent line of research of
``mechanism design from samples'', see, e.g., \cite{GW21, BSV16, CR14, DHP16, MR16, GHZ19}.

The learning problem can be stated as follows. Given a joint distribution $\textbf{V}$ of the values of the sellers and buyers, where our
access to the distribution is by getting random samples from it, we now aim to find a \nice mechanism with the maximum expected gains from
trade over the distribution $\textbf{V}$.  
Specifically, we want to learn such a mechanism that, with high probability,  approximates the optimal gains from trade that is obtainable by a \nice mechanism for the distribution $\textbf{V}$, 
to within an additive $\varepsilon$. 
Can this be done from samples? and if so, how many sample points from $\textbf{V}$ do we need?

\subsection*{Starting point: bilateral trade}
Let us start with the simplest trade setting, that of bilateral trade. 
For that setting the learning problem is well understood. 
In a bilateral trade setting there is a single buyer and single seller with an item to sell, 
and their private values for the item are $v_s$ and $v_b$, respectively. 
These pairs of values $(v_s,v_b)$ are jointly distributed according to a distribution $\textbf{V}$. %
The goal is to find a \nice mechanism with maximum GFT over all \nice mechanisms. 
The starting point for addressing this type of question is a characterization of the family
of possible \nice  mechanisms, and then optimize over that family.  
Luckily, \nice mechanisms for bilateral trade have a very limited form: they are fixed price mechanisms. That is, any  \nice mechanism is defined by a fixed price $p$, and trade occurs, at price $p$, if $v_s < p < v_b$, and no trade happens if $v_s> p$ or $p> v_b$ (in case of ties, trade may or may not happen). 
It follows that the challenge in this case is to learn a price $p$ such that the fixed-price mechanism with 
price $p$ maximizes the gains from trade for $\textbf{V}$ (over all such fixed-price mechanisms).
Since this requires learning just a single parameter, $p$, this is indeed
doable and, as implied by \cite{Cesa-BianchiCCF21}\footnote{This paper actually solves a harder problem, in a regret-minimization
setting, as well as analyzes several related models. See Section \ref{sec:related} for details.}, 
choosing $O(\varepsilon^{-1})$ sample points and 
picking the {\em sampled value}\footnote{It is interesting to note, though, that it does not suffice to choose among some
predetermined $\delta$-grid of values.}
that yields highest GFT {\em on the sample} gives, with high probability, a mechanism that is optimal up to an additive 
$\varepsilon$ on the real, yet unknown,  distribution $\textbf{V}$.
Significantly, this holds for any joint distribution  $\textbf{V}$, %
and does not require the buyer's and seller's values
to be independent.

We thus see that the problem is solved and is relatively simple for bilateral trade settings. Is this also the case for more involved trade settings?

\subsection*{Beyond  bilateral trade: $2$ buyers and $1$ seller}

We saw that for bilateral trade it is possible to learn a good \nice mechanism from samples. Can this be done when we have more than a single buyer and a single seller?
To address this we consider the next most-simple trade setting, that of one seller of a single item, and \emph{two} buyers (rather than just one).
In the 1-seller 2-buyer setting, the values of the seller for the item $v_s$, the values of the two buyers are $v_1$ and $v_2$, respectively, and the  triplet $(v_s,v_1,v_2)$ is sampled from a joint distribution $\textbf{V}$.

Following the blueprint to the %
learning problem taken for bilateral trade, we first aim to understand the space of \nice mechanisms and the trade allocations that they
provide. It turns out that the class of \nice mechanisms for the 1-seller 2-buyer case can be essentially characterized by a pair of single-parameter functions. 

\begin{definition}\label{def:compat}
    A pair $(f_1,f_2)$ of functions $f_1, f_2 : \mathbb{R}_+ \rightarrow \mathbb{R}_+ \cup \{\infty\}$ is called {\em compatible},
    if for no pair of values  $(v_1,v_2)$ it holds that $v_1 > f_1(v_2)$ and $v_2 > f_2(v_1)$.  A mechanism $M$ is 
    \emph{associated with a pair of compatible functions $(f_1, f_2)$},
    if when $v_s < f_1(v_2) < v_1$ it sells
    the item to player 1 at price $f_1(v_2)$; when  $v_s < f_2(v_1) < v_2$ it sells the item to player 2 at price $f_2(v_1)$; and there is no trade  if 
    $v_s > \max\{f_1(v_2), f_2(v_1)\}$, or if $v_1 < f_1(v_2)$ and $v_2 < f_2(v_1)$.
\end{definition} 

We show that every \nice mechanism has a unique   compatible pair of functions that is associated to it, and, conversely, every pair of compatible functions is associated with a \nice mechanism (Proposition \ref{prop:mech-pair}).  A compatible pair of functions does not completely define a mechanism as  it leaves some flexibility about the outcome in case of ties. We identify a specific tie breaking rule that always yields the
highest gains from trade for a given compatible pair of functions (see Definition \ref{def:m}) and denote the
associated mechanism under that tie breaking rule by $M_{f_1,f_2}$.  We illustrate the concept of a pair of compatible functions in Figure \ref{fig:randomexample}. The pair of compatible functions that is associated with the GFT-maximizing \nice mechanism for the uniform distribution over $[0,1]^3$ (denoted by $U[0,1]^3$) is illustrated in Figure \ref{fig:uniform}.

\begin{figure}[htp]
\centering
\begin{minipage}{.45\textwidth}
  \centering
  \includegraphics[width=\linewidth]{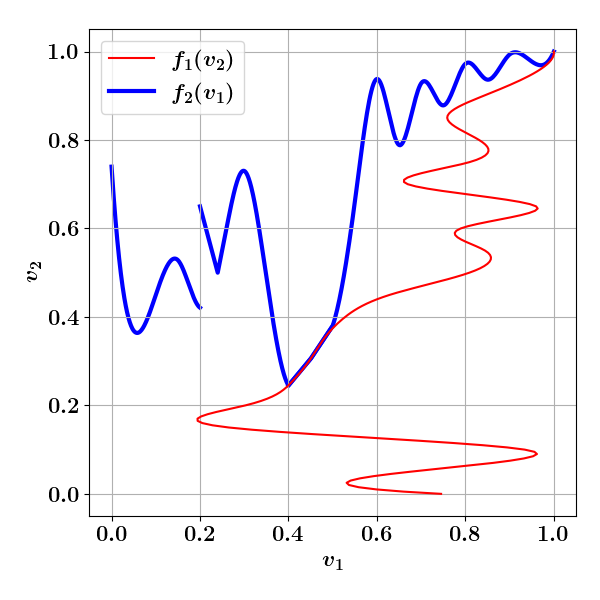} 
  \caption{Example of a pair $(f_1, f_2)$ of compatible functions. Such a pair satisfies that there is no point that is strictly above $f_2(v_1)$ and strictly on the right of $f_1(v_2)$.}
  \label{fig:randomexample}
\end{minipage}%
\hfill
\begin{minipage}{.45\textwidth}
  \centering
  \includegraphics[width=\linewidth]{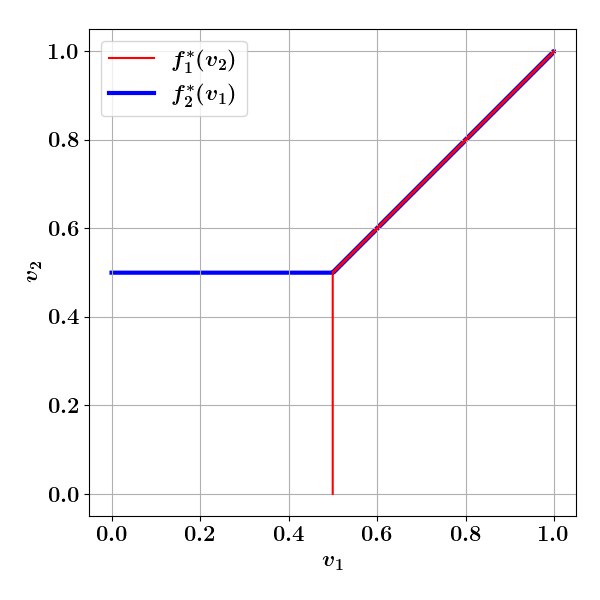} 
  \caption{The pair of compatible functions that is associated with the GFT-maximizing \nice mechanism for the uniform distribution over $[0,1]^3$.}
  \label{fig:uniform}
\end{minipage}
\end{figure}

\subsection*{Arbitrary joint distributions: a negative result}

Given the characterization above, it is natural to
attempt learning a (nearly) optimal \nice mechanism by optimizing over all mechanisms 
$M_{f_1,f_2}$, over all  pairs of compatible functions $(f_1, f_2)$.
Unfortunately, this class turns out to have infinite ``dimension'', and
in fact, we show
that no low-dimension subclass suffices for even approximating
the gains from trade.  The impossibility result is very general and shows that
it is even impossible to learn a ``general'' mechanism in the sense of
\cite{myerson1983efficient} that need only be 
Bayesian incentive compatible,  interim individually rational, 
and weakly budget balanced.

\begin{restatable}{theorem}{nocortri}\label{thm:no-cor-tri}
Consider the $1$-seller $2$-buyer setting. 
   There exists a constant $c>0$ such that for every finite $t$, 
    there does not exist an algorithm that accepts $t$ random samples from any unknown distribution $\textbf{V}$ on $[0,1]^3$ and, 
    with high probability, outputs a \nice mechanism (or even any general mechanism) 
    with gains from trade that is within an additive $c$ from the 
    maximum gains from trade obtainable by a \nice mechanism.

\end{restatable}

This is in strong contrast to the positive result for bilateral trade mentioned
above, where $O(1/\varepsilon)$ samples are sufficient for learning an approximately-optimal \nice mechanism. This negative result implies that a slightly more general setting (with two buyers rather than one) is significantly more challenging for learning, and the results for this setting are fundamentally different than the results for bilateral trade.   

This impossibility result is very general.  First, it
immediately extends to %
any double-auctions setting beyond bilateral trade: In any setting with more than two buyers the result follows from making all but one seller and two buyers irrelevant for trade\footnote{That is, buyers have the minimal value of $0$, and sellers have the maximal value of $1$.}. If there is only one buyer but multiple sellers, a symmetric argument proves a similar result.  We also note that since the impossibility is shown for an additive error 
when values are bounded by $1$, it also holds for multiplicative approximation.

\OLD{\mbc{OLD:} We are able to prove the impossibility result not only for  1-seller 2-buyer markets, but for every market beyond bilateral trade. %

\mbc{a new version of the negative result for \nice mechanisms:}

}

Our proof shows that it is hard to distinguish between the following two 
extreme cases: the first where a \nice mechanism obtains the first-best (optimal)
gains from trade\footnote{The first-best GFT is the expected optimum GFT (expected realized GFT under the given distribution).}, and the second where even a general mechanism does not approximate the first best.  
Let us briefly outline the core observation behind  this proof. 
Let  $S= (s_1, s_2, \ldots, s_k)= \{ (v_s^j,v_1^j,v_2^j)\}_{j=1}^k$ be a set of triplets of values of the seller and the two buyers.   We say that $S$ is \emph{generic} 
if no value appears more than once in it, and let 
$U(S)$ be the uniform distribution over $S$.

\begin{lemma}\label{lem:fb-uniform-sample}
Let $S= (s_1, s_2, \ldots, s_k)= \{ (v_s^j,v_1^j,v_2^j)\}_{j=1}^k$ be a 
generic (multi-)set of size $k$ over $[0,1]^3$.
Consider a 1-seller 2-buyer setting with agents' values  samples from $U(S)$, the uniform distribution over $S$. 
Then there exists a \nice mechanism with GFT that equals the first-best GFT on $U(S)$. 
\end{lemma}

The \nice mechanism that extracts the first-best GFT can be viewed as a perfect over-fitting of the mechanism to the $k$ triplets in $S$. 
number of samples $t$ it is impossible to distinguish between the following two cases: 
We show that this simple lemma implies the theorem, by proving that for any finite
number of samples $t$, it is impossible to distinguish between the following two cases:
1) the $t$ samples are coming from the uniform distribution over $[0,1]^3$, for which 
there is a gap between the first best and what a strategic mechanism can achieve, and
2) the $t$ samples are from $U(S)$ for a ``random'' generic set $S$ of size $k>>t$, for which there is no such gap.  
This impossibility is due to the fact that for large enough $k$, it is unlikely that we sample any element of $S$ more than once.\footnote{Note that for the bilateral-trade setting an appropriate adjustment of the claim of Lemma \ref{lem:fb-uniform-sample} is not true (consistently with the fact that Theorem \ref{thm:no-cor-tri} does not hold for bilateral trade). 
That is, consider the set $S$ with two samples $(s,b)$ taking the realizations $(0,2)$ and $(5,7)$ (note that this set of values is generic).  As any \nice mechanism for bilateral trade is a fixed-price mechanism, the maximum expected GFT of any \nice mechanism on a uniform distribution over these two samples is $1$, while the first-best GFT is $2$.}

\subsection*{The independent case: a learning algorithm}

Given our strong negative result for arbitrary joint distributions, we now focus our attention on the special case where the valuations of the seller and two buyers are sampled independently, that is, the private values $(v_s,v_1,v_2)$ are sampled from the product distribution $V_s \times V_1 \times V_2$.  
First, we tighten the partial characterization of optimal \nice mechanisms. While, as we saw, the class of \nice mechanisms is too rich to learn, we are able to identify a smaller subclass of \nice mechanisms that for every three independent distributions is guaranteed to include a \nice mechanism that has at least as much GFT as any other \nice mechanism. %
While \nice mechanisms can be associated with  arbitrary compatible pairs of functions, we show that for product distributions it is enough to restrict our search to pair of \emph{monotone} functions that are \emph{tight}. 

\begin{definition}\label{def:tight}
    A pair $(f_1,f_2)$ of functions $f_1,f_2: \mathbb{R}_{\geq 0} \rightarrow \mathbb{R}_{\geq 0}$ is called \em{tight after} $(p_1,p_2)$ if there
    do {\em not} exist any values $v_1 \geq p_1, v_2 \geq p_2$ such that 
    $v_1 < f_1(v_2)$ and $v_2 < f_2(v_1)$.     The functions are called \textup{tight} if there exists a point $(p_1,p_2)$ after which they are tight.
\end{definition}

Figure \ref{fig:uniform} illustrates the pair of compatible functions that is associated with the GFT-maximizing \nice mechanism for the distribution $U[0,1]^3$.
Observe that this pair of functions is tight, with $p_1=p_2=\frac{1}{2}$, and, additionally, the functions are monotone.   
Our main characterization result is 
that {\em tight} compatible pairs of {\em monotone} functions always suffice for 
maximizing gains from trade in {\em the independent case}.

\begin{restatable}{theorem}{bestmech}\label{thm:bestmech}
    For any three distributions $V_s,V_1,V_2$ and any \nice mechanism $M$, there exists a compatible and tight pair $(f_1^*, f_2^*)$ of monotone functions such that the GFT of the 
     \nice mechanism $M_{f_1^*,f_2^*}$ (as per Proposition \ref{prop:mech-pair}) on $\textbf{V}= V_s\times V_1\times V_2$
     is at least as large as the GFT of $M$ on $\textbf{V}$. 
\end{restatable}

For a product distribution $V$, we denote the optimal gains from trade of 
a \nice mechanism by $\OPTniceV$. %
A tight compatible pair of monotone functions that is associated with such a mechanism is called an \emph{optimal} pair for the distribution. Figures \ref{fig:uhalf} and \ref{fig:complexexample}
demonstrate that optimal pairs may be far from trivial.
Nevertheless, as we will show, the space of tight compatible pairs of
monotone functions has a small enough dimension as to allow efficient learning, where
the critical element is the monotonicity of the functions.

\begin{figure}[htp] 
\centering
\begin{minipage}{.45\textwidth}
  \centering
  \includegraphics[width=\linewidth]{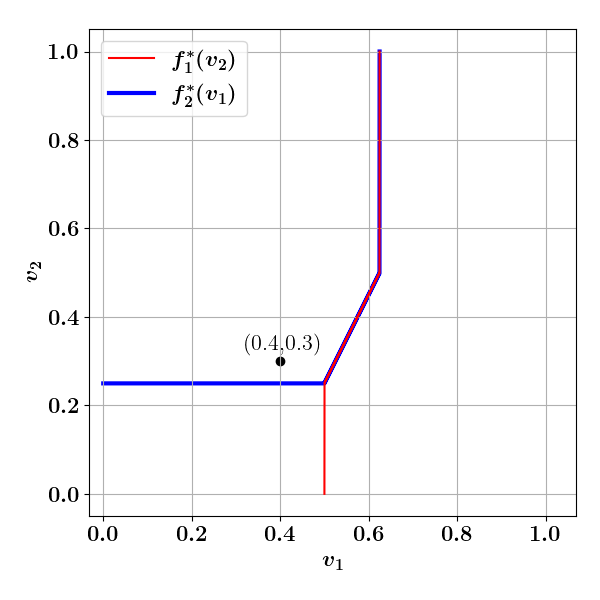} 
  \caption{Optimal pair of functions for 
  the independent distributions $v_s\sim U[0,1]$, $v_1\sim U[0,1]$, $v_2\sim U[0,\frac{1}{2}]$. Note that the buyer with lower value sometimes trades, e.g., for values $(v_s,v_1,v_2)= (0, 0.4,0.3)$.}
  \label{fig:uhalf}
\end{minipage}%
\hfill
\begin{minipage}{.45\textwidth}
  \centering
  \includegraphics[width=\linewidth]{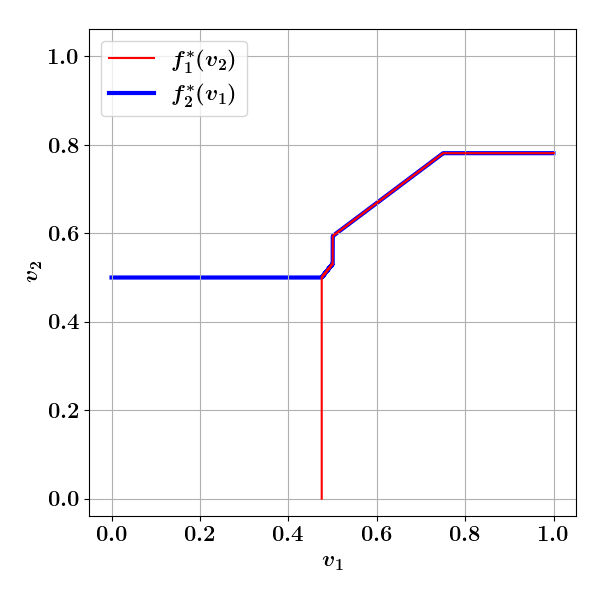} 
   \caption{Optimal pair for $v_s\sim U[0,1]$, \\$v_1\sim \begin{cases} U[0,\frac{1}{2}], & \frac{1}{4} \\ U[\frac{1}{2}, \frac{3}{4}], & \frac{3}{4} \end{cases}$, $v_2\sim \begin{cases} U[0,\frac{1}{4}], & \frac{3}{4} \\ U[\frac{1}{2}, 1], & \frac{1}{4} \end{cases}$. \\ Note that the functions are rather complicated.}
  \label{fig:complexexample}
\end{minipage}
\end{figure}

Before continuing with our goal of learning from samples,
we first present an algorithm that can find an optimal pair for a finitely supported and
{\em explicitly given} product distribution, i.e. where the distribution
over each player's value is 
given as a finite list of (value, probability) pairs.  
Such an algorithm is not trivial and we obtain it using careful dynamic programming.  

\begin{restatable}{theorem}{algmax}\label{thm:algmax}
    There exists a polynomial time algorithm that when explicitly given finite distributions $V_s, V_1, V_2$,  
    outputs a compatible pair of functions $(f_1^*,f_2^*)$ such that $M_{f_1^*,f_2^*}$ maximizes gains from trade for $\textbf{V}= V_s\times V_1\times V_2$ over all
    \nice mechanisms.
\end{restatable}

To learn a good mechanism for a product distribution from a given sample, we use this optimization algorithm over the sample.  I.e. our learning algorithm first draws a sample of
appropriate size and then finds the GFT-optimal mechanism for the empirical distribution. Our main result shows that, with high probability, this mechanism is almost optimal 
for the true (unknown) product distribution as well.  The critical element that drives this theorem is the 
low dimension of the class of tight and compatible pairs of monotone functions, which as Theorem \ref{thm:bestmech} shows,
are sufficient to consider in order to optimize gains from trade over all \nice mechanisms in the case
of product distributions.

\begin{restatable}{theorem}{thmlearn}\label{thm:learn}
    Let $V_s, V_1, V_2$ be distributions over $[0,1]$, 
    and fix $\varepsilon, \delta > 0$. Denote the maximum expected GFT that can be achieved by a \nice mechanism on $\textbf{V}=V_s\times V_1 \times V_2$ by $\OPTniceV$. There exists a polynomial-time algorithm that given sample access to $V_s, V_1, V_2$ and parameters $\varepsilon, \delta$, 
    outputs, with probability at least $1-\delta$, a mechanism which has an expected GFT  
    of at least $\OPTniceV - \varepsilon$ on $\textbf{V}$, using $poly(\frac{1}{\varepsilon}\log \frac{1}{\varepsilon\delta})$ samples from $\textbf{V}$.
\end{restatable}

Note that our learning result is for bounded distributions. Such an assumption is necessary, as if the support is unbounded, almost all the gains from trade might come from extremely rare events that are never sampled, and thus a good mechanism cannot be learned.\footnote{As usual, learning results may be also recovered under assumptions
that appropriately bound the tail of the distribution.}

\subsection*{Open Problems}

Our %
results in this paper were focused on \nice mechanisms for the case of one seller and
two buyers.  It would be natural to extend these results along any one of two separate dimensions.
First, rather than consider the ``nicest'' class of \nice mechanisms that are deterministic, dominant-strategy incentive-compatible, ex-post individual-rational, and
strong budget-balanced, it would be interesting to consider more general classes 
of mechanisms that relax some (or all) of
these assumptions: to randomized, Bayesian Incentive-compatible,
interim individual rationality, and weak budget-balance, respectively.  Second,
it is natural to consider the case of more than a single seller and two buyers: most
generally, to the case of $k$ buyers and $l$ sellers, or at least to the case of
$k$ buyers and a single seller.

As mentioned above, our impossibility result does indeed generalize in both dimensions:
it holds both for any number of buyers and sellers (as long as it is  beyond the bilateral case) and shows that learning algorithms that are even allowed to use the most
general type of mechanisms cannot approximate the benchmark set by even
the most restricted type of mechanisms.

Our positive result, however, is very specific and it seems that any significant extension
along either of these two dimensions may require new ideas.  We leave the challenge
of extending our results along any one of these dimensions (or showing that this
would be impossible) as our main set of open problems.

In terms of numbers
of sellers and buyers, we simply do not have any sufficiently good characterization 
of gains-from-trade maximizing mechanisms for any case with more than one seller and two
buyers, even for independent valuations.  I.e. we lack an analog of
Theorem \ref{thm:bestmech} that will allow such mechanisms to be learned.  
For example we do not even know how to solve the learning problem in the case
of one seller and three buyers with independent valuations.

In terms of the class of mechanisms considered, again we 
lack sufficiently good characterizations for any richer class of mechanisms\footnote{
Except for the case of allowing universally-truthful randomized mechanisms rather
than deterministic ones, which one may observe does not improve the gains-from-trade.}
and thus we do not know whether richer classes of mechanism have sufficiently 
low dimension to be efficiently learned.
For example, we do not even know whether relaxing the strong budget-balance requirement
to weak budget-balance may improve the gains-from-trade, and if so, whether it is possible to learn such mechanisms in the case of one seller and two buyers with
independent valuations.  

When considering Bayesian Incentive-Compatibility (BIC), one
should be careful in the definition of the learning challenge. This is so as the exact notion
of BIC may require
an exact specification of the distribution which would not be available from sample data,
so it might be helpful to focus on the weaker notion of $\varepsilon$-BIC.  On the other hand, 
the algorithmic question itself may be easier.  E.g., finding the optimal
Bayesian Incentive-Compatible mechanism for a given sample can be solved using a 
linear program.

\subsection{Additional related work}\label{sec:related}

Prior work \cite{Cesa-BianchiCCF21} has studied the bilateral trade problem in a regret-minimization framework over rounds of seller/buyer interactions, with no prior knowledge on the private seller/buyer valuations.
They showed that learning a close to optimal \nice mechanism can be done in some challenging online scenarios (but provably not in some even more challenging ones). We study the easier problem of (offline) learning from samples, but the  setting of 1-seller 2-buyers and beyond (while they study bilateral trade).

The problem of gains-from-trade maximization in trade settings has been extensively studied. Myerson and Satterthwaite \cite{myerson1983efficient} 
show that for bilateral trade settings, the second-best mechanism  has strictly lower GFT than the first best. The second-best mechanism is the mechanism with the highest GFT subject to being  Bayesian Incentive-Compatible (BIC), Interim Individually Rational (IIR), and ex-ante Weakly Budget Balanced (WBB).
In the double auctions setting, McAfee\cite{McAfee92} has shown that there is a DSIC mechanisms that is IIR and WBB with GFT that is at least $(1-1/q)$ fraction of the first best, when $q$ is the size of the efficient trade. This result was extended to other trade settings,  e.g., to 
spatially-distributed markets \cite{BNP09}.
Note that the approximation provided by this mechanism to the first best is $0$ when $q=1$, and thus this mechanism gives no approximation in small markets, and in particular, in single-seller  settings (e.g., bilateral trade and the 1-seller 2-buyer setting we focus on). %
Recently, it was shown that for bilateral trade (and for double auctions), the multiplicative gap between the first and second best  is only a constant factor \cite{DengMSW22}. 

For double auctions, \cite{BCGZ18} has presented mechanisms that are constant approximation to the second best, and are also asymptotically efficient. 
In \cite{BGG2020}, a {Bulow-Klemperer}-style result for GFT  maximization was presented to the double auctions setting, showing that with IID agents, a variant of the McAfee\cite{McAfee92} mechanism 
with an additional buyer ensures at least as much expected GFT as the first best of the original setting. 

A long line of research has studied constant-fraction approximation to the welfare in bilateral trade \cite{BlumrosenD21, CW23,KPV2022} and more general trade settings \cite{CGKLTS2020}.